\documentclass[journal,twoside]{IEEEtran}
\usepackage{amsmath,amsthm,amsfonts,amssymb,latexsym,extarrows,booktabs,array,arydshln}
\usepackage{fancyhdr,graphicx,tabularx}
\usepackage{multirow}
\usepackage{makecell}
\usepackage{amsfonts}
\usepackage{graphicx}
\usepackage{subfigure}
\usepackage {cite}
\usepackage {url}
\usepackage {stfloats}
\usepackage {mathrsfs}
\usepackage{algorithm}
\usepackage{algorithmic}
\usepackage{xcolor}
\usepackage{longtable}
\usepackage[OT2,T1]{fontenc}

\IEEEoverridecommandlockouts
\newtheorem{definition}{Definition}
\newtheorem{proposition}{Proposition}

\newtheorem{assumption}{Assumption}
\newtheorem{problem}{Problem}
\newtheorem{theorem}{Theorem}
\newtheorem{remark}{Remark}
\newtheorem{fact}{Fact}
\newtheorem{question}{Question}
\begin{document}
 \date{}
\title{Independence-Checking Coding for OFDM  Channel Training  Authentication:  Protocol Design, Security, Stability, and  Tradeoff Analysis}
\author{Dongyang~Xu,~\IEEEmembership{Student Member,~IEEE,}
        Pinyi~Ren,~\IEEEmembership{Member,~IEEE,}
        and James~A.~Ritcey,~\IEEEmembership{Fellow,~IEEE}

}
\maketitle
\vspace{-40pt}
\begin{abstract}
In  wireless OFDM communications systems,  pilot tones, due to their \emph{publicly-known} and \emph{deterministic} characteristic,  suffer  significant  jamming/nulling/spoofing risks. Thus, the convectional  channel training  protocol using pilot tones  could be attacked and  paralysed, which raises the issue of anti-attack channel training  authentication (CTA), that is, verifying  the claims of  identities of pilot tones and channel estimation samples.  In this paper, we consider  one-ring scattering scenarios with large-scale uniform linear arrays (ULA) and  develop an independence-checking coding (ICC) theory to build a secure and stable  CTA protocol, namely ICC based CTA (ICC-CTA) protocol.  In this protocol,  pilot tones are  not merely  \emph{randomized} and inserted into  subcarriers, but also  encoded as diversified subcarrier activation patterns (SAPs) simultaneously. Those encoded SAPs, though camouflaged by malicious signals, can be  identified and decoded into original pilots, and hence for high-accuracy channel impulse response (CIR) estimation. The CTA security is first characterised by the error probability of  identifying legitimate CIR estimation samples. We prove that the identification error probability (IEP) is equal to zero under  the continuously-distributed mean angle of arrival (AoA) and  also derive a closed-form expression of IEP under the  discretely-distributed case. The CTA instability is  formulated  as  the function of probability  of stably estimating  CIR  against  all available  diversified  SAPs.  A realistic  tradeoff  between the CTA  security and instability  under the discretely-distributed AoA is identified and  an optimally-stable  tradeoff problem is  formulated, with the objective of   optimizing the code rate to maximize  security  while  maintaining maximum stability  for ever.  Solving this, we derive the closed-form expression of  optimal code rate. Numerical results finally validate the resilience of proposed  ICC-CTA protocol.
\end{abstract}
\IEEEpeerreviewmaketitle
\section{Introduction}
\label{introduction}
\IEEEPARstart{W}{ith} the evolution of air interface  towards 5G, security paradigms for the protection of  air interface technologies  have attracted increasing attentions in wireless communications systems. Safeguarding  the current standard, for instance, orthogonal frequency-division multiplexing (OFDM)  or securely implementing  the initiation, such as massive antenna technique,  gradually come up on the  agenda~\cite{Bogale}. The common problem encountered  is that the imperishable characteristic of wireless channels, such as the open and shared nature, has always been rendering those  technologies vulnerable to the growing  denial of service (DoS) attacks~\cite{Yan}. A phenomenon, if we notice,  has emerged in the  physical (PHY) layer  that   DoS attacks, with  moderate  size of the involved network segment and modest  implementation complexity, have become increasingly common and potent~\cite{Rahbari}. As their major hacking behaviors, radio jamming (RJ) attacks  have  been exhibiting its  astonishing destructive power on those existing~\cite{Shahriar} and  emerging air interface  techniques~\cite{Pirzadeh}.

Among these RJ attacks, protocol-aware attack serves as the most effective one as the attacker could sense the specific protocols and intensify its  effectiveness significantly  by jamming a physical layer mechanism instead of data payload directly~\cite{Litchman}. The typical case which frequently occurs in massive-antenna OFDM systems is that protocol-aware attackers always show a great  appetite for the  channel training  protocol. In this protocol, frequency-domain subcarrier (FS) channels and  channel impulse response (CIR) samples, are estimated to  further the  high-quality user experience using those estimations.  The  motivations for this case are twofold. On one hand, multi-antenna OFDM technique has been deployed universally in current commercial and military applications,  which  incurs huge interests of  malicious nodes. Since the channel training protocol requires that   \emph{deterministic and publicly-known} pilot tones  should be shared on the time-frequency resource grid (TFRG) by  all parties~\cite{Lichtman}, a pilot-aware attacker could sense and acquire  the public  pilot information, and  practically behave  in such a way that  the regular channel training process may not be maintained as usual~\cite{Clancy2,Sodagari,Xu}. On the other hand, everyone has witnessed  the introduction of  massive  antennas  into OFDM technique which has been  promoted  significantly in  the recent practice, such as in 3GPP new radio (NR) specifications.  In this era,  the precise channel training  becomes very crucial to  maintaining the significant multiplexing gains of target users. The bad news is that  imprecise estimation samples  could not only lower down  those gains  but also benefit others, such as the attacker, due to the high resolution of antenna arrays. What's more, when the channel training is misguided in favour of attacker, actually without too much efforts,  massive antenna arrays in OFDM systems will be well loved  by the attacker.

In this context, authenticating channel training  becomes  very critical to the massive antenna  OFDM systems since it determines the authenticity of  channel estimation results.  Generally, channel training launched by any certain  subscriber is  authenticated by default through the designated  public pilot tones  allocated to that subscriber~\cite{Tu, Lai}.  Applying  the same  pilot tones as the subscriber at the receiver  to   channel estimation  means  the exact authentication for channel training. This process is called the channel training authentication (CTA) which belongs  to the field of physical-layer authentication~\cite{Yu}. Intrinsically, exact  CTA mainly depends on  the authenticity of  pilot tones in a sense that  the claims of identities of pilot tones should be  verified.   The uniqueness and non-reproducibility of  pilot tones  are two foremost requirements which however will no longer hold true when  a  pilot-aware attacker  jams/nulls/spoofs  those  pilot tones.  In practice,  attacking  CTA process in OFDM systems  is a common phenomenon, e.g., in  scenarios with tactical consideration~\cite{Peng} or in Long Term Evolution (LTE)-based  public safety networks~\cite{Doumi}. Those attacks, including pilot tone jamming (PTJ) attack~\cite{Clancy2}, pilot tone nulling (PTN) attack~\cite{Sodagari} and pilot tone spoofing (PTS) attack~\cite{Xu},  are very hard to eliminate  once they have occurred successfully.

\subsection{Related Works}
Much of the work related to securing CTA  has been investigated thus far.  How to detect the alteration to authenticity and how to protect and further maintain the high authenticity are two major branches in this area.

The first attempt for  narrow-band single-carrier systems is made in~\cite{Zhou}  in which  the pilot contamination (PC) attack, one type  of  PTS attack,  was introduced and evaluated. Following~\cite{Zhou},  much of the work was studied,  but  limited to the  detection of authenticity of pilot signals by exploiting the physical layer information, such as auxiliary training or data sequences~\cite{Kapetanovic1,Tugnait2,Xiong1} and some prior-known channel information~\cite{Kapetanovic2,Kang}.  Different from those, authors in~\cite{Adhikary} first studied the advantage  of spatial correlation in  the maintenance of  authenticity of pilots,  and found that the natural spatial separation of massive antenna arrays can force PC attack to occur effectively only in a particular angular domain. However, we should never forget that the attacker is out of control. In this regard, PC attack actually becomes more well-directed, rather than less effective.

The first attempt for multi-subcarrier scenarios was presented  by Clancy et al.~\cite{Clancy1}, verifying  the  possibility and effectiveness of PTJ attack. Following this,   PTJ attack was then studied for  single-input single-output (SISO)-OFDM communications in~\cite{Clancy2} which also introduced the PTN attack and then extended it to the  multiple-input multiple-output (MIMO)-OFDM system~\cite{Sodagari}.  The initial  attempt  to resolve pilot aware attack  for  conventional OFDM systems was proposed  in~\cite{Shahriar2}, that is, transforming the PTN and PTS attack into   PTJ attack by  randomizing the  locations and values of regular pilot tones on time-frequency resource grid (TFRG).   It  figured out the importance that  pilot tone scheduling, even being random, would also affect  channel acquisition. Hinted by this,  authors in~\cite{Xu} proposed a FS channel estimation framework under the PTS attack by exploiting pilot randomization  and the independence component analysis (ICA) theory. One key problem is that the practical subcarriers are not mutually independent in the scenarios with limited channel taps, and thus ICA does not apply  in this case. Most importantly,  the CIR estimation is impossible. Basically, CIR  is very critical to the CTA in future 5G mobile eco-systems in which  measuring the multipath before designing  systems is mandatory since the channel has to carry the big amount of data for our ``everything wireless'' applications.  The knowledge of the channel response  represents the aggregate values of gross physical multipath information. CIR is such a wideband channel characterization and contains all information necessary to simulate or analyze any type of radio transmission through the channel.  For instance, the amplitude of channel taps could reflect the sparsity of channel in some cases and their variations could tell us the Doppler spread, coherence bandwidth, and so forth~\cite{Tse}.

To solve those issues, our previous work  in~\cite{Xu2}  proposed an independence-checking coding (ICC)  method which  provides high authenticity guarantee on the FS channel and  CIR estimation based on randomized pilot tones. Nevertheless, the influence of randomization  on CIR estimation was not evaluated and optimized, which incurs the instability of CIR estimation. In this sense,  CTA not only merely  requires the high security against attacks, but also strongly and necessarily calls for the  high stability of CIR estimation accuracy.  As far as we know,  there were very few studies jointly considering the  security  and instability during the  channel training phase.
 \begin{figure}
\centering \includegraphics[width=1.0\linewidth]{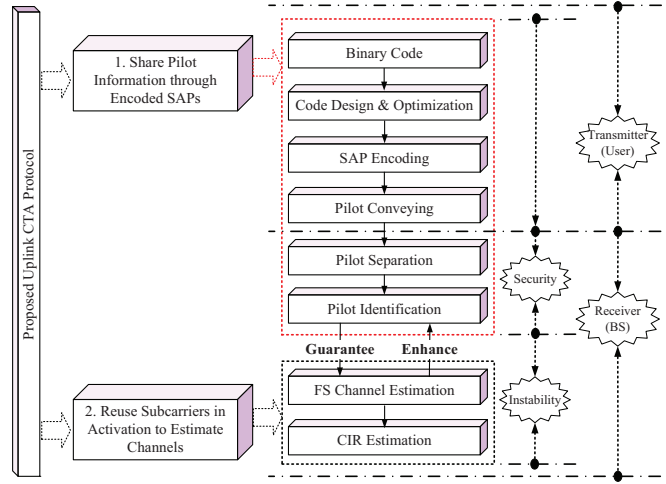}
\caption{  Design methodology for  CTA protocol  in the uplink; Pilot conveying, separation and identification are achieved in Section~\ref{IUTA}.  Channel estimation and identification enhancement are realized in Section~\ref{FSCE}. The tradeoff issue between security and instability is given in Section~\ref{RST}. }
\vspace{-10pt}
\label{Methodology}
\end{figure}
\subsection{Motivations and Contributions}
The hints from the above investigation   further motive us to build up a secure CTA  protocol for massive-antenna OFDM systems with  considerations of  the  heterogeneity of attack modes and the  instability of CIR estimation

Recall that pilot randomization  serves as a commonsense technique for defending against pilot-aware attack.  However, inserting randomized pilot tones on TFRG solely  functions to transform the attack modes such that the attack issue will not be insolvable, rather than to resolve the issue practically. To be more specific, this  brings  two  bottlenecks, i.e., \textbf{1) Unpredictable attack modes;}
\begin{problem}[\textbf{Attack Model}]
A pilot-aware attacker chooses on TFRG a hybrid attack mode including  PTJ  attack and silence cheating (SC) mode.   In  PTJ  attack mode, two  behaviors are available, i. e., wide-band  pilot jamming (WB-PJ) attack~\cite{Yan2} and partial-band pilot  jamming (PB-PJ) attack~\cite{Abdelhakim}. In SC mode,  the  attacker  keeps silent  for cheating  the legitimate node.  The legitimate node can never acquire the behaviors of the attacker in advance. All of the three modes can be very effective due to the node transparency (i.e., no association or independent with each other) and should  never be ignored.
\end{problem}

\textbf{2) Irreversible pilot information. } Randomized pilot information  become  irreversible in the following sense:
\begin{problem}
Randomized pilot information are naturally  camouflaged by random channel information. Those information, if transmitted by pilot tones for uplink channel training   through wireless channels, cannot be separated and identified.
\end{problem}
This problem inspires us to perform the \textbf{protocol design} for  the overall channel training process.  The guideline for this is presented  in Fig.~\ref{Methodology}  where  two key requirements are detailed as follows:

\begin{enumerate}
\item \textbf{Share pilot information through encoded  subcarrier activation
patterns (SAPs):} Selectively activate and deactivate  OFDM subcarriers by transmitting pilots on subcarriers or not, and create various SAP candidates. Encode  all SAPs  as a binary code.  Optimize the code set  in such a way  that  arbitrary  one SAP, namely, codeword, if  suffering a hybrid  attack in the wireless environment, are enabled to be separated and identified securely.  With this preparation,  pilot information is conveyed and encoded as one codeword and  further expressed   as a SAP.  Secure  pilot sharing is thus constructed  between  transceiver pairs.
\item \textbf{Reuse  subcarriers in activation to estimate channels:}  Generate channel estimators according to the identified pilots and apply them  on the activated subcarriers for FS channel  estimation.  Enhance the pilot identification  using the estimated FS channels. Derive CIR estimation samples from the estimated FS channels.
\end{enumerate}
In this  methodology, channel estimation coexists with the information coding and  the two techniques influence each other.
In spite of  the security guarantee provided by  encoded SAPs,   SAP diversification also incurs the uncertainties as to the  amount  and distribution of  subcarriers in activation,  further instabilizing  the CIR estimation  extremely.  This   entanglement between security and instability   motivates us  to perform the \textbf{protocol optimization}. The main contributions of this paper are  summarized as follows:
\begin{enumerate}
\item \textbf{Protocol Design:} First, we establish a fundamental   principle for encoding  arbitrary SAPs as a binary code set precisely. Following this, we develop an  ICC theory to further optimize the code such that  arbitrary two codewords in the code, if being superimposed on each other, can be  separated and identified securely.  In order  to evaluate the security for this, we formulate  two key performance indicators (KPIs), i.e.,  the separation error probability (SEP) and identification error probability (IEP).  We prove that  SEP is always guaranteed to be zero and also  derive the analytical  expression of IEP.  We   build up an uplink  ICC based CTA (ICC-CTA) protocol  in which  legitimate  transceiver pair   encodes and decodes randomized pilot phases securely   through the ICC codebook, and then performs FS channel and CIR estimation using the identified pilots.
\item  Next, we discover  a hidden  phenomenon that  when  FS channel estimation is performed on the basis of this protocol,  the array spatial correlation  existing in the  overlapping subcarriers that also carry information  from both  the legitimate node and the attacker can further help reduce IEP in one-ring scattering scenarios. At this point, the attacker can actually help the legitimate node to enhance  the security.  Interestingly, it can be proved that  zero IEP cannot be achieved only when  the attacker is located in  the  clusters with the same  mean   angle of arrival (AoA) as the legitimate node. This principle, in this sense, could facilitate   the acquisition of  the position of attacker.  Theoretically when we consider the mean AoA with  continuous probability distribution, the security, in theory, can be perfectly guaranteed. Practically in discretely-distributed case, we  give an analytical expression of  how much  the security  could be further improved.
 \item  \textbf{Protocol Optimization:} Finally, we identify the phenomenon of  instable CIR  estimation  in  this protocol  and define the stability by the function of  probability of stable CIR estimation  against diversified SAPs.  In the realistic scenario with discretely-distributed mean AoAs,  we  identify and model  the  tradeoff between the security and instability.  Interestingly, we prove that there always  exists   an optimally-stable tradeoff for which  the CIR estimation can always achieve its optimal stability without losing estimation precision asymptotically. Maintaining this stability,  we  further determine a closed-form expression of optimal code rate that maximizes the security. This code rate indicates how to  flexibly configure  the number of activated subcarriers  under this hybrid attack   such that desirable security and  maximum stability of CIR estimation can be both guaranteed.
\end{enumerate}

\emph{Organization:} In Section~\ref{PSA} , we present an overview of  pilot-aware attack on  massive-antenna OFDM  systems. In Section~\ref{IUTA}, we introduce an ICC-CTA protocol.   FS channel estimation and security enhancement  are described  in  Section~\ref{FSCE}.  Security-instability tradeoff  in CIR estimation  is  provided in Section~\ref{RST}.  Numerical results are presented in Section~\ref{NR} and finally  we conclude our work in Section~\ref{Conclusions}.

\emph{Notations:} We use boldface capital letters ${\bf{A}}$ for matrices, boldface small letters ${\bf{a}}$ for vectors , and small letters $a$ for scalars. ${{\bf{A}}^*}$, ${{\bf{A}}^{\rm{T}}}$, ${{\bf{A}}^{{H}}}$  and ${\bf{A}}\left( {:,1:L} \right)$ respectively denotes   the  conjugate operation, the transpose,  the conjugate transpose  and the first $L$ columns of matrix ${\bf{A}}$.  $\left\| {\cdot} \right\|$ denotes the Euclidean norm of a vector or a matrix. $\left| {\cdot} \right|$ is the cardinality of a set.  ${\mathbb{E}}\left\{  \cdot  \right\}$ is the expectation operator. $\otimes$ denotes  the Kronecker   product operator.   ${\rm{diag}}\left\{ {\bf{a}} \right\}$   stands for the diagonal matrix with  the elements of column vector $\bf{a}$ on its diagonal.
\section{Overview of  Pilot-Aware Attack on   Massive-Antenna OFDM  Systems}
\label{PSA}
We in this section  outline a fundamental  overview of  CTA  issue under pilot aware attack, from a mathematical point of  view.  This refers to  the basic system model, signal model, and channel estimation model.  Finally, the   pilot randomization technique is described and most importantly, we identify  its  potential challenges in resolving the attack.
\begin{figure}[!t]
\centering \vspace{-10pt}\includegraphics[width=1.0\linewidth]{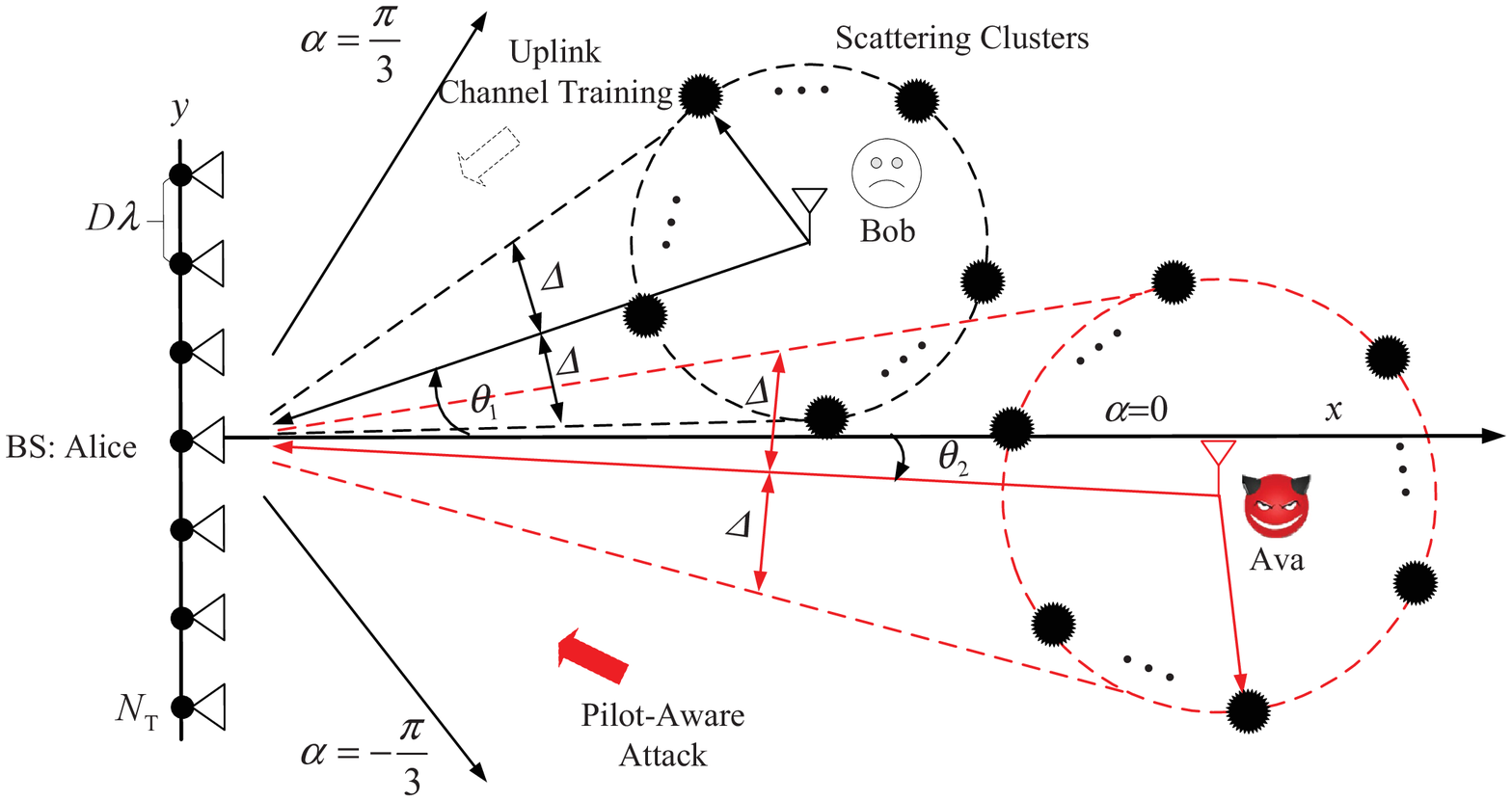}
\caption{Diagram of large-scale MISO-OFDM system under the  wide-band one-ring scattering  model. In this system,  AoA ranges of  Bob and Ava overlap with each other, which incurs an effective pilot-aware attack on the  uplink channel estimation.}\label{fig:system_b}
\vspace{-20pt}
\label{System_model}
\end{figure}
\subsection{System Description}
\begin{table*}\footnotesize
\begin{center}
\caption{Summary of Notations.}
\begin{tabular}{|l|l|}
\hline
Notations & Description \\
\hline
$N_{T}$; $D\lambda~(0 \le D \le {1 \mathord{\left/
 {\vphantom {1 2}} \right.
 \kern-\nulldelimiterspace} 2})$  & Number of antennas at BS; Antenna spacing\\
 ${\Delta }$; $\theta_{i}, i=1,2$ & Angle spread at  BS; Mean AoA of Bob, $i=1$ and Ava, $i=2$ \\
$\overline N$; $N$; $\left( {\overline N\ge N} \right)$  &   Total available number of subcarriers within each OFDM symbol time; Length of FFT points\\
$ N_{\rm B};N_{\rm A}$ $\left( {N_{\rm B}\le \overline N, N_{\rm A}\le \overline N }\right)$      &   Number of subcarriers allocated  for Bob and Ava\\
$\Psi {\rm{ = }}\left\{ {{\rm{0}},{\rm{1}} \ldots {\overline N}{\rm{ - 1}}} \right\}$ & Index set of total available subcarriers\\
${\Psi _{\rm{B}}}{\rm{ = }}\left\{ {{i_0},{i_1}, \ldots ,{i_{{N_{\rm B}} - 1}}} \right\}$, ${\Psi _{\rm{A}}}{\rm{ = }}\left\{ {{i_0},{i_1}, \ldots ,{i_{{N_{\rm A}} - 1}}} \right\}$ & Index set of  subcarriers allocated  for Bob and Ava \\
${x_{\rm{B}}^j\left[ k \right]}, {j \in {\Psi _{\rm{B}}}}$;  ${x_{\rm{A}}^j\left[ k \right]}, {j \in {\Psi _{\rm{A}}}}$ &  Pilot tones for Bob and Ava  at the $j$-th subcarrier and $k$-th symbol time\\
$\rho _{{\rm{B}}}$, $\rho _{\rm{A}}$; ${\phi_{k}}$, $\varphi_{k,i}$ & Uplink  training power for Bob and Ava; Pilot  phases of Bob and Ava  \\
$L$; $\sigma ^2$ &  Number of sampled multi-path taps in baseband, Average noise power of BS\\
${\bf{h}}_{\rm{B}}^i\in {{\mathbb C}^{L \times 1}} $; ${\bf{h}}_{\rm{A}}^i \in {{\mathbb C}^{L\times 1}}$  &  CIR  vectors, respectively from Bob and Ava to the $i$-th receive antenna of Alice \\
${\bf {F}}\in {{\mathbb C}^{N \times N}}$;  $\bf {F_{\rm L}}$; ${{\bf F}_{{\rm L},k}}$; ${{\bf  F}_{j}}$ & DFT matrix;   ${{\bf{F}}_{\rm{L}}} = \sqrt N {\bf{F}}\left( {:,1:L} \right)$;  $k$-row matrix of $ \bf {F_{\rm L}}$; $j$-row matrix of $\bf {F}$.\\
 ${{\bf{v}}^i}\left[ k \right]\in {{\mathbb C}^{N \times 1}}$,  ${{\bf{v}}^i}\left[ k \right] \sim {\cal C}{\cal N}\left( {0,{{{\bf{I}}_N}\sigma ^2}} \right)$  &  AWGN  vector  at the $i$-th antenna of BS within the $k$-th symbol time \\
${{\bf{w}}^i_{j}}\left[ k \right]={{\bf{F}}_{j}}{{\bf{v}}^i}\left[ k \right]$, $1\le j\le N$  &  AWGN vector across $j$ subcarriers for  $i$-th antenna of BS within $k$-th symbol  \\
 $\sigma _{{\rm{B}},l}^2$; $\sigma _{{\rm{A}},l}^2$ &  PDP  of  the $l$-th path of  Bob and Ava\\
 ${{\bf{y}}_i}\left[ k \right]$  &   Received signal vector  at the $i$-th subcarrier  and $k$-th OFDM symbol. \\
 ${\cal A}$; & $\left\{ {{\phi}:{{{\phi} = 2m\pi } \mathord{\left/
 {\vphantom {{{\phi _k} = 2m\pi } C}} \right.
 \kern-\nulldelimiterspace} C},0 \le m \le C - 1} \right\}$;  $C$ denotes  the quantization resolution  \\
 ${\cal P}_{d}=\left\{ {k_1, \ldots ,k_d} \right\}$, $ {{\cal P}_{d}} \subseteq \Psi $ & Index set of ambiguous subcarriers  under hybrid attack\\
${\cal P}_{s}=\left\{ {j_1, \ldots ,j_s} \right\}$, $ {{\cal P}_{s}} \subseteq \Psi, \left| {\cal P}_{s} \right| = s$ & Index set of overlapping subcarriers under hybrid attack \\
${\cal P}_{a}=\left\{ {i_1, \ldots ,i_a} \right\}$, $ {{\cal P}_{a}} \subseteq \left\{ {1, \ldots {N_{\rm{T}}}} \right\}$, $\left| {\cal P}_{a} \right| = a$ & Index set of the intersection  of ${\cal S}_1$ with ${\cal S}_2$\\
${{\bf{R}}_{{i}}}\in {{\mathbb C}^{N_{\rm T} \times N_{\rm T}}}$;  ${\bf{R}}_{{\rm F}}$ &  Channel covariance matrix of Bob ($i=1$) and Ava ($i=2$);  ${\bf{R}}_{{\rm F}} {\rm{=}}  {\bf{F}}_{{\rm{L}},{s}}^{\rm{T}}{\bf{F}}_{{\rm{L}},{s}}^{\rm{*}}$\\
 $\rho_{i}$; ${\rho _{{\rm{f}}}} = \min \left\{ {s,L} \right\}$ & Rank of ${{\bf{R}}_i}$; Rank of  ${{\bf{R}} _{\rm{F}}}$\\
 $N_1^{\rm{d}}$;  $N_0^{\rm{d}}$ &  Total number of  non-zero digits in \textbf{S.1} and zero digits in  \textbf{S.2} \\
$N_{1,i}^{\rm s}$; $N_{0,i}^{\rm s}$, $i=0,1$ &  Total number of  nonzero digits for ${{\cal A}_i}$; Total number of zero digits for  ${{\cal A}_i}$\\
${d_{i_{\rm r}}}$& Digit indicated by  RS\\
\hline
\end{tabular}
\end{center}
\vspace{-15pt}
\end{table*}
We consider a  synchronous large-scale  multiple-input single-output (MISO)-OFDM system with a $N_{\rm T}\gg 1$-antenna base station (named as Alice) and a single-antenna legitimate user (named as Bob). As shown in  Fig.~\ref{System_model},  the based station (BS)  is equipped with a $D\lambda$-spacing directive  uniform linear array (ULA) and  placed at the origin along the $y$-axis to serve a 120-degree sector that is centered around the $x$-axis ($\alpha  = 0$). We assume no energy is received for angles $\alpha  \notin \left[ {-\frac{\pi }{3},\frac{{\pi }}{3}} \right]$.  The summary of notations  is given  in Table I.

For a typical cellular configuration, the channel from Bob to Alice is a correlated random vector with covariance matrix that depends on the scattering geometry. Assuming a macro-cellular tower-mounted BS with no significant local scattering, the propagation between Bob and Alice   is characterized by the local scattering around Bob, resulting in the well-known one-ring model~\cite{Adhikary}.  For OFDM systems with frequency-selective channels, the wide-band configuration is more realistic. Here, we consider the wide-band one-ring scattering  model in which   Bob  is  surrounded by local scatterers  within  $\left[ {{\theta _1} - {\Delta },{\theta _1} + {\Delta }} \right]$~\cite{Adhikary,Han}.This will contribute to  the following mathematical characterisation  of the  advantage of spatial correlation in security provision as an explicit result, rather than a complex and  unintuitive implication.

We consider pilot tone based uplink channel training process  on time-frequency domain with $\overline N$ available subcarriers at each OFDM symbol time. In principle,  subcarriers  indexed by  ${\Psi _{\rm{B}}}$ are employed  for  pilot tone insertion and  the following channel estimation. Those pilot tones, known as reference signals in  LTE-A and/or beyond, are  deterministic and publicly-known in  TFRG.  Each transceiver, by sharing those tones, can deduce the FS channels and further estimate  the CIR samples.  Therefore a single-antenna malicious node (named as Ava)  could disturb  this training process  by jamming/spoofing/nulling  those pilot tones. We denote  the set of victim  subcarriers  by ${\Psi _{\rm{A}}}$ and   make the following assumption:
\begin{assumption}
Ava is  surrounded by local scatterers  within $\left[ {{\theta _2} - {\Delta },{\theta _2} + {\Delta }} \right]$ and always has  common or overlapping AoA intervals with Bob, this is, $\left[ {{\theta _2} - \Delta ,{\theta _2} + \Delta } \right] \cap \left[ {{\theta _1} - \Delta ,{\theta _1} + \Delta } \right] \ne \emptyset $
\end{assumption}
This assumption is supported by the scenario where a common large scattering body (e.g., a large building) could create a set of angles common to all nodes in the system. In this case, the angular spread of BS is broad and  the overlapping of AoA intervals is inevitable. The result is that   the  channel covariance  eigenspaces of Bob and Eva are coupled and  the attack is hard to  eliminate through angular separation~\cite{Adhikary}.
\begin{assumption}
We consider the multiple-cluster scenario. Two types of the distribution model of  ${\theta _i}$, $i=1,2$ are considered, including the continuous probability distribution (CPD)~\cite{Adhikary} and the discrete probability distribution (DPD)~\cite{Bhagavatula}, for instance, discrete uniform distribution with the support of interval length $K$.
\vspace{-10pt}
\end{assumption}
\subsection{Receiving Signal Model}
In this subsection, we introduce the receiving  signal model at Alice.  To begin with, we will give the concept of  pilot insertion pattern (PIP) which indicates  the way of  inserting   pilot tones  across  subcarriers and OFDM symbols.
\begin{assumption}[\textbf{Frequency-domain PIP}]
We in this paper assume $x_{{\rm{B}}}^i\left[ k \right] = {x_{{\rm{B}}}}\left[ k \right]= \sqrt {{\rho _{{\rm{B}}}}} {e^{j{\phi_{k} }}}$, ${{i \in {\Psi _{\rm{B}}}}}$ for low overhead consideration and theoretical analysis. Alternatively,  we can superimpose ${x_{{\rm{B}}}}\left[ k \right] $ onto a dedicated  pilot sequence optimized under a non-security oriented scenario and utilize this new pilot for training. At this point,  $\phi_{k}$ can be an additional phase difference for security consideration. We do not impose the phase constraint on the PIP  strategies  of  Ava, that is,  ${x^i_{\rm{A}}}\left[ k \right] = \sqrt {{\rho_{\rm{A}}}} {e^{j{{\varphi}_{k,i}} }},{{i \in {\Psi _{\rm{A}}}}}$.
\end{assumption}

Let us proceed to the basic OFDM procedure. First,  the frequency-domain pilot signals of Bob and Ava  over $N$ subcarriers  are  respectively  stacked as $N$ by $1$  vectors ${{\bf{x}}_{\rm{B}}}\left[ k \right] = \left[ {{x_{{\rm{B}},j}}\left[ k \right]} \right]_{{ {j \in {\Psi}}}}^{\rm{T}}$ and ${{\bf{x}}_{\rm{A}}}\left[ k \right] = \left[ {{x_{{\rm{A}},j}}\left[ k \right]} \right]_{{ {j \in {\Psi}}}}^{\rm{T}}$. Here  there exist:
  \begin{equation}\label{E.1}
{x_{{\rm{B}},j}}\left[ k \right] = \left\{ {\begin{array}{*{20}{c}}
{x_{\rm{B}}\left[ k \right]}&{{j \in {\Psi _{\rm{B}}}}}\\
0&{{j \notin{\Psi _{\rm{B}}}}}
\end{array}} \right., {x_{{\rm{A}},j}}\left[ k \right] = \left\{ {\begin{array}{*{20}{c}}
{x_{\rm{A}}^j\left[ k \right]}&{{j \in {\Psi _{\rm{A}}}}}\\
0&{{j \notin{\Psi _{\rm{A}}}}}
\end{array}} \right.
\end{equation}
Assume that the  length of cyclic prefix is larger than  $L$.  The parallel streams, i.e.,  ${{\bf{x}}_{{\rm{B}}}}\left[ k \right]$ and  ${{\bf{x}}_{{\rm{A}}}}\left[ k \right]$, are modulated with inverse fast Fourier transform (IFFT). After removing the cyclic prefix at the $i$-th receive antenna and $k$-th OFDM symbol time, Alice derive the time-domain $N$ by $1$  vector ${{\bf{y}}^i}\left[ k \right]$  as:
\begin{equation}\label{E.3}
{{\bf{y}}^i}\left[ k \right] = {\bf{H}}_{{\rm{C,B}}}^i{{\bf{F}}^{\rm{H}}}{{\bf{x}}_{\rm{B}}}\left[ k \right] + {\bf{H}}_{{\rm{C,A}}}^i{{\bf{F}}^{\rm{H}}}{{\bf{x}}_{\rm{A}}}\left[ k \right] + {{\bf{v}}^i}\left[ k \right]
\end{equation}
where ${\bf{H}}_{{\rm{C,B}}}^i$ and ${\bf{H}}_{{\rm{C,A}}}^i$ are $N \times N$ circulant matrices for which the  first column of  ${\bf{H}}_{{\rm{C,B}}}^i$  and ${\bf{H}}_{{\rm{C,A}}}^i$ are respectively given by ${\left[ {\begin{array}{*{20}{c}}
{{\bf{h}}_{\rm{B}}^{{i^{\rm{T}}}}}&{{{\bf{0}}_{1 \times \left( {N - L} \right)}}}
\end{array}} \right]^{\rm{T}}}$ and ${\left[ {\begin{array}{*{20}{c}}
{{\bf{h}}_{\rm{A}}^{{i^{\rm{T}}}}}&{{{\bf{0}}_{1 \times \left( {N - L} \right)}}}
\end{array}} \right]^{\rm{T}}}$. Here, ${\bf{h}}_{\rm{A}}^i $ is   assumed to be independent with  ${\bf{h}}_{\rm{B}}^i$.
Taking   fast Fourier transform (FFT),  Alice  finally derives the    frequency-domain $N$ by $1$ signal vector  at the $i$-th receive antenna and $k$-th OFDM symbol time as
\begin{equation}\label{E.4}
{\widetilde {\bf{y}}^i}\left[ k \right] = {\rm{diag}}\left\{ {{{\bf{x}}_{\rm{B}}}\left[ k \right]} \right\}{{\bf{F}}_{\rm{L}}}{\bf{h}}_{\rm{B}}^{{i}}  + {\rm{diag}}\left\{ {{{\bf{x}}_{\rm{A}}}\left[ k \right]} \right\}{{\bf{F}}_{\rm{L}}}{\bf{h}}_{\rm{A}}^{{i}} + {{\bf{w}}^i_{N}}\left[ k \right]
\end{equation}
Throughout this paper, we assume that the CIRs belonging to  different paths  at each antenna  exhibit spatially uncorrelated Rayleigh fading.
Without loss of generality, each path has the uniform and normalized power delay profile (PDP) satisfying  $\sum\limits_{l = 1}^L {\sigma _{{\rm B},l}^2}  = 1, \sum\limits_{l = 1}^L {\sigma _{{\rm A},l}^2}  = 1$~\cite{McKay}.  For each path, CIRs  of different antennas  are spatially correlated.
 With the one-ring scattering model, the correlation  between  channel coefficients of antennas $1 \le m,n \le N_{\rm T}$, $\forall l$ is defined by~\cite{Adhikary,Han}:
    \begin{equation}\label{E.5}
{\left[ {{{\bf{R}}_{{k}}}} \right]_{m,n}} = \frac{1}{{2\Delta L}}\int_{ - {\Delta} + {\theta_{k}}}^{{\Delta} + {\theta_{k}}} {{e^{ - j2\pi D\left( {m - n} \right)\sin \left( \theta  \right)}}} d\theta, k=1,2
  \end{equation}
Here, ${{\bf{R}}_{{k}}}, k=1,2$ are  symmetric positive semi-definite matrices.  Note that  ${{\bf{R}}_{{2}}}$ is unknown for Alice and Bob while ${{\bf{R}}_{{1}}}$  is known by Alice.
\subsection{ Channel Estimation Model}
For the PTS attack, Ava could learn  the pilot tones employed by Bob  in advance and  impersonate Bob  by utilizing the same pilot tone learned. There exists $
{\Psi _{\rm{B}}}\cup {\Psi _{\rm{A}}} = {\Psi _{\rm{B}}}$ and $x_{{\rm{A}}}^i\left[ k \right] = {x_{{\rm{B}}}}\left[ k \right],{i \in {\Psi _{\rm{A}}}}$. Signals in  Eq.~(\ref{E.4}) can be rewritten as:
\begin{equation}\label{E.6}
{\widetilde {\bf{y}}^i_{\rm PTS}}\left[ k \right] = {{\bf{F}}_{\rm{L}}}{\bf{h}}_{\rm{B}}^{{i}}{x_{\rm{B}}}\left[ k \right]  + {{\bf{F}}_{\rm{L}}}{\bf{h}}_{\rm{A}}^{{i}}{x_{\rm{B}}}\left[ k \right] + {{\bf{w}}^i_{N}}\left[ k \right]
\end{equation}

Finally, a least square (LS)  based channel estimation is formulated by the equation
${\widehat {{\bf{h}}}_{con}^i} = {\bf{h}}_{\rm{B}}^i + {\bf{h}}_{\rm{A}}^i + {\left( {{{\bf{F}}_{\rm{L}}}} \right)^ + }\frac{{x_{\rm{B}}^{\rm{H}}\left[ k \right]}}{{{{\left| {x_{\rm{B}}^{\rm{H}}\left[ k \right]} \right|}^2}}}{{\bf{w}}^i_{N}}\left[ k \right]$
where $\left( {{{\bf{F}}_{\rm{L}}}} \right)^+ $ is  the Moore-Penrose
pseudoinverse of $ {{{\bf{F}}_{\rm{L}}}} $.  We  see that   the estimation of ${\bf{h}}_{\rm{B}}^i $  is contaminated by  ${\bf{h}}_{\rm{A}}^i $   with  a noise bias when a PTS  attack happens. As to the characterisation of  PTN  attack and PTJ attack, we can refer to   the mathematical interpretation  in~\cite{Xu2}.
\subsection{Influence of Pilot Randomization on   Pilot-Aware Attack}
Pilot randomization can avoid the pilot aware attack without imposing  any prior information on the pilot design. The common method is to randomly  select  phase candidates.  Each of the phase candidates is   mapped by default into a unique quantized  sample, chosen  from the set ${\cal A}$.
 Since phase information only provides  the security guarantee as shown in Assumption 3, thus without the need of huge overheads,   we make the following assumptions:
 \begin{assumption}[\textbf{Time-domain PIP}]
 During  two adjacent OFDM symbol time, such as, $k_i,k_{i+1}$, $i\ge0$ two pilot phases ${\phi _{{k_i}}}$ and ${\phi _{{k_{i+1}}}}$  are kept with fixed phase difference, that is,   ${\phi _{{k_{i+1}}}} - {\phi _{{k_{i}}}} = \overline \phi $, for reducing the authentication overheads. Here, ${\phi _{{k_{i+1}}}}$ and ${\phi _{{k_{i}}}}$ are both random but $\overline \phi$  are deterministic and  publicly known.
 \end{assumption}
Institutively,  how the value $C$ increases affects  the  performance of   anti-attack technique. This technique also brings up the subject of Problem 2.

 \section{ICC-CTA Protocol}
 \label{IUTA}
As shown in the Fig.~\ref{Methodology}, this section presents the principles of  pilot conveying,  separation and  identification.
\subsection{Pilot Conveying  on Code-Frequency Domain}
\label{PCBC}
Basically, the more phases supported in $\cal A$, the higher  coding diversity is required, and thus the more available SAPs  should be created. Theoretically, this requires a delicately-designed binary code  and  practically depends on how to  activate and deactivate  subcarriers as the code indicates. This operation will inevitably induce a concurrence of activated and deactivated subcarriers,  and therefore detecting the number of signals coexisting on one subcarrier is a necessary work before coding.

To achieve this goal, we will  employ  the technique of eigenvalue ratio based detection (ERD)  proposed in~\cite{Shakir}. Here  we consider three symbol time and  a $3 \times N_{\rm T}$ receiving signal matrix, denoted by ${{\bf{Y}}_{\rm{D}}}$,  is created for detection. Given the normalized covariance matrix defined by $\widehat {\bf{R}} = \frac{1}{{{\sigma ^2}}}{\bf{Y}}_{\rm D}{{\bf{Y}_{\rm D}}^{\rm{H}}}$, we define  its ordered eigenvalues  by  ${\lambda _{{1}}} > {\lambda _2}  > {\lambda _3} > 0$ and construct  the test statistics by $T = \frac{{{\lambda _{1}}}}{{{\lambda _{3}}}}\mathop {\gtrless}\limits_{{{{\overline{\cal H}}_0}}}^{{{{{ {\cal H}}_0}}}} \gamma$ where $\gamma$ denotes the  decision threshold. The hypothesis ${{ {\cal H}}_0}$  means that there exist signals and  ${{\overline{\cal H}}_{0}}$ means the opposite.
\subsubsection{Construction of Code Frequency Domain}
Given the threshold  $\gamma$,  the cumulative distribution function (CDF) of $T$, denoted by $F\left( \gamma  \right)$,  can be expressed by $F\left( \gamma  \right) =1-{P_f}= \Phi \left\{ {\frac{{{\zeta _{{\lambda _3}}}\gamma  - {\zeta _{{\lambda _1}}}}}{{{\xi _{{\lambda _1}}}{\xi _{{\lambda _3}}}\chi \left( \gamma  \right)}}} \right\},     \chi \left( \gamma \right) = \sqrt {\frac{{{\gamma^2}}}{{\xi _{{\lambda _1}}^2}} - \frac{{2\rho \gamma }}{{{\xi _{{\lambda _1}}}{\xi _{{\lambda _3}}}}} + \frac{1}{{\xi _{{\lambda _3}}^2}}}$ where $\rho  = {{\left( {{\zeta _{{\lambda _1},{\lambda _3}}} - {\zeta _{{\lambda _1}}}{\zeta _{{\lambda _3}}}} \right)} \mathord{\left/
 {\vphantom {{\left( {{\zeta _{{\lambda _1},{\lambda _3}}} - {\zeta _{{\lambda _1}}}{\zeta _{{\lambda _3}}}} \right)} {{\xi _{{\lambda _1}}}{\xi _{{\lambda _3}}}}}} \right.
 \kern-\nulldelimiterspace} {{\xi _{{\lambda _1}}}{\xi _{{\lambda _3}}}}}$~\cite{Shakir} . Here $\Phi\left\{ \cdot \right\}$ denotes  CDF of a standard Gaussian random variable.  In order to measure how many  antennas   are required on each subcarrier to achieve a certain  ${P_f}$, a decision threshold  function $\gamma  \buildrel \Delta \over =f\left( {{N_{\rm{T}}},{P_f}} \right)$ is derived, with $f\left( {{N_{\rm{T}}},{P_f}} \right) =  \frac{{{\zeta _{{\lambda _1}}}{\zeta _{{\lambda _3}}} - \tau ^2{\rho }{\xi _{{\lambda _1}}}{\xi _{{\lambda _3}}} + {\tau }\sqrt {{\delta } - 2{\rho }{\xi _{{\lambda _1}}}{\xi _{{\lambda _3}}}{\zeta _{{\lambda _1}}}{\zeta _{{\lambda _3}}}} }}{{\zeta _{{\lambda _3}}^2 - \tau ^2\xi _{{\lambda _3}}^2}}$
 where ${\delta } = \zeta _{{\lambda _1}}^2\xi _{{\lambda _3}}^2 + \zeta _{{\lambda _3}}^2\xi _{{\lambda _1}}^2 + \left( {\rho ^2 - 1} \right)\tau ^2\xi _{{\lambda _1}}^2\xi _{{\lambda _3}}^2$, $\tau  = {\Phi ^{ - 1}} \left\{1-{P_f}\right\}$.
Here $\zeta_{\lambda_i^{p}}={\mathbb{E}}\left( {\lambda _i^p} \right),i=1,3, p=1,2$, $\zeta_{\lambda_1\lambda_3}={\mathbb{E}}\left( {\lambda_1\lambda_3} \right)$ and $\xi _{{\lambda _i}}^2 = {\mathbb E}\left( {\lambda _i^2} \right) - {\left[ {{\mathbb E}\left( {{\lambda _i}} \right)} \right]^2}, i=1,3$.  The related parameters can be shown as follows:
\begin{equation}
\zeta_{\lambda_1^{p}}{\rm{ = }}C_{{N_{\rm{T}}},{\rm{3}}}^{{\rm{ - 1}}}\sum\limits_{i,j = 1}^3 {{{\left( { - 1} \right)}^{i + j}}} 2\Gamma \left( {{L_{{\alpha _1},1}}} \right)\Gamma \left( {{L_{{\alpha _2},2}}} \right)G_{i,j}
\end{equation}
where there exists $
G_{i,j} = \sum\limits_{{l_1} = 1}^{{L_{{\alpha _1},1}} - 1} {\sum\limits_{{l_2} = 1}^{{L_{{\alpha _2},2}} - 1} {\frac{{\Gamma \left( {{l_1} + {l_2} + {p_{i,j}} - 1} \right)}}{{{l_1}!{l_2}!{3^{{l_1} + {l_2} + {p_{i,j}} - 1}}}}} }  - \sum\limits_{{l_1} = 1}^{{L_{{\alpha _1},1}} - 1} {\frac{{\Gamma \left( {{l_1} + {p_{i,j}} - 1} \right)}}{{{l_1}!{2^{{l_1} + {p_{i,j}} - 1}}}}}  - \sum\limits_{{l_2} = 1}^{{L_{{\alpha _2},2}} - 1} {\frac{{\Gamma \left( {{l_2} + {p_{i,j}} - 1} \right)}}{{{l_2}!{2^{{l_2} + {p_{i,j}} - 1}}}} + \Gamma \left( {{p_{i,j}} - 1} \right)}.
$
\begin{equation}
\zeta_{\lambda_3^{p}}{\rm{ = }}C_{{N_{\rm{T}}}{\rm{3}}}^{{\rm{ - 1}}}\sum\limits_{i,j = 1}^3 {{{\left( { - 1} \right)}^{i + j}}} 2\Gamma \left( {{L_{{\alpha _1},1}}} \right)\Gamma \left( {{L_{{\alpha _1},2}}} \right)G_{i,j}^1
\end{equation}
where $G_{i,j}^1 = \sum\limits_{{l_1} = 1}^{{L_{{\alpha _1},1}} - 1} {\sum\limits_{{l_2} = 1}^{{L_{{\alpha _2},2}} - 1} {\frac{{\Gamma \left( {{l_1} + {l_2} + {p_{i,j}} - 1} \right)}}{{{l_1}!{l_2}!{3^{{l_1} + {l_2} + {p_{i,j}} - 1}}}}} } $.
\begin{equation}
\zeta_{\lambda_1,\lambda_3}{\rm{ = }}C_{{N_{\rm{T}}}{\rm{3}}}^{{\rm{ - 1}}}\sum\limits_{{i_1},{i_3},{j_1},{j_3}} \chi  \Gamma \left( {{L_{{\beta _1},1}}} \right)\left\{ {\sum\limits_{{l_1} = 1}^{{L_{{\alpha _1},1}} - 1} {\frac{{G_{i,j}^2}}{{{l_1}!}}} } \right\}
\end{equation}
where $\chi  = {\left( { - 1} \right)^{{i_1} + {i_3} + {j_1} + {j_3}}}$ and we have $G_{i,j}^2= \sum\limits_{{l_1} = 1}^{{L_{{\alpha _1},1}} - 1} {\frac{{\Gamma \left( {{l_1} + p_{i,j}^3 + 1} \right)}}{{{2^{{l_1} + p_{i,j}^3 + 1}}}}} \left\{ {\Gamma \left( {p_{i,j}^1 + 1} \right) - \sum\limits_{t = 0}^{p_{i,j}^3 + {l_1}} {\frac{{{2^t}\Gamma \left( {t + p_{i,j}^1 + 1} \right)}}{{t!{3^{t + p_{i,j}^1 + 1}}}}} } \right\} + \sum\limits_{{l_2} = 1}^{{L_{{\alpha _2},2}} - 1} {\Gamma \left( {p_{i,j}^3 + 1} \right)} \left\{ {\frac{{\Gamma \left( {{l_1} + p_{i,j}^1 + 1} \right)}}{{{2^{{l_1} + p_{i,j}^1 + 1}}}} - \sum\limits_{t = 0}^{p_{i,j}^3} {\frac{{\Gamma \left( {{l_1} + t + p_{i,j}^1 + 1} \right)}}{{t!{3^{{l_1} + t + p_{i,j}^1 + 1}}}}} } \right\}$. For the parameters therein, there exist ${C_{{N_{\rm{T}}},3}} = 2\prod\limits_{i = 1}^3 {\left( {{N_{\rm{T}}} - i} \right)!} $,  ${p_{i,j}} = {N_{\rm{T}}} + p + i + j - 3,p_{i,j}^1 = {N_{\rm{T}}} + p + {i_1} + {j_1} - 5,p_{i,j}^3 = {N_{\rm{T}}} + p + {i_3} + {j_3} - 5$, ${\alpha _1} = 2,{\alpha _2} = 1$,  ${L_{{\alpha _k},k}} = \left\{ {\begin{array}{*{20}{c}}
{{N_{\rm{T}}} - 4 + k + {\alpha _k}}&{{\alpha _k} < i,k < j}\\
{{N_{\rm{T}}} - 2 + k + {\alpha _k}}&{{\alpha _k} \ge i,k \ge j}\\
{{N_{\rm{T}}} - 3 + k + {\alpha _k}}&{otherwise}
\end{array}} \right.$ and ${L_{{\beta _k},k}} = \left\{ {\begin{array}{*{20}{c}}
{{N_{\rm{T}}} - 4 + k + {\beta _k}}&{{\beta _k} < ,k < }\\
{{N_{\rm{T}}} - 3 + k + {\beta _k}}&{ < {\beta _k} < \bar i,k < ,or, < {\beta _k} < \bar j,k < }\\
{{N_{\rm{T}}} - 1 + k + {\beta _k}}&{ < {\beta _k} < \bar i,k > \bar j,or, < {\beta _k} < \bar j,k > \bar i}\\
{{N_{\rm{T}}} + k + {\beta _k}}&{{\beta _k} > \bar i,k > \bar j}\\
{{N_{\rm{T}}} - 2 + k + {\beta _k}}&{otherwise}
\end{array}} \right.$.  ${\mathop{\rm sgn}} (\cdot)$ is the Signum function and $\Gamma \left( {\cdot} \right)$ is the upper incomplete Gamma function.
 \begin{figure}[!t]
\centering \includegraphics[width=0.85\linewidth]{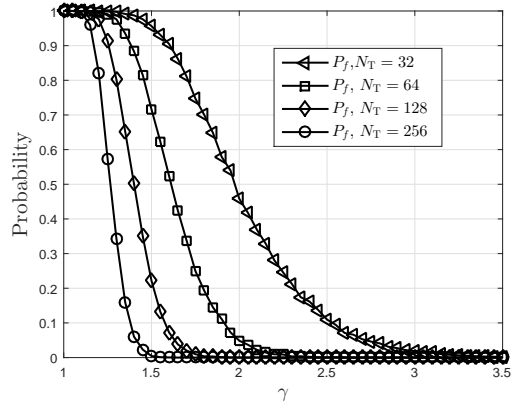}
\caption{ Performance of ERD on   single subcarrier. }
\vspace{-10pt}
\label{PD_PF}
\end{figure}

Using the expression of $\gamma$, we establish a single-subcarrier encoding (SSE) principle  to  encode the number of detected signals  into binary digits, i.e, 0 or 1.
 \begin{definition}[\textbf{SSE Principle}]
 One subcarrier  can be precisely encoded if, for any $\varepsilon > 0$, there exists a positive number $\gamma \left( \varepsilon  \right)$  such that, for all  $\gamma  \ge \gamma \left( \varepsilon  \right)$, ${P_f}$  is smaller than $\varepsilon $.
\end{definition}

Based on the Definition 1, we can encode  the $ m$-th subcarrier as  a binary  digit   ${s_{ m}} $   according to
${s_{m}} = \left\{ {\begin{array}{*{20}{c}}
1&\rm {{{\cal H}_0}~is~ true }\\
0&{otherwise}
\end{array}} \right.$.
We should note that $f\left( {{N_{\rm{T}}},{P_f}} \right)$ is a monotone decreasing  function of two independent variables, i.e., ${N_{\rm{T}}}$ and ${P_f}$. For a given probability constraint ${\varepsilon}^*$, we could always expect a lower bound $\gamma \left( {{\varepsilon ^*}} \right)$ such that $\gamma \left( {{\varepsilon ^*}} \right) = f\left( {{N_{\rm{T}}},{\varepsilon ^*}} \right)$ is satisfied. Under this equation, we could flexibly configure  ${N_{\rm{T}}}$ and $\gamma \left( {{\varepsilon ^*}} \right)$ to make ${\varepsilon}^*$  approach zero~\cite{Shakir}.   We also find that  the value of $\gamma$  achieving zero-${P_f}$  is decreased with the increase of  ${N_{\rm{T}}}$.

To verify this, we consider  three OFDM symbols  and flexible configuration of $N_{\rm T}$, such as, from $32$ to $256$. We simulate ${P_f}$ against   various $\gamma$ in Fig.~\ref{PD_PF}.  As we can see,   the required decision threshold $\gamma$ is decreased with the increase of the number of antennas. This fact also further  verifies the feasibility of Definition 1. For example, we can find a desirable point at $\gamma=1.5, N_{\rm T}=256$ where  $P_f$  is equal to zero,  thus facilitating  perfect binary coding for each kind of SAPs.

 \begin{figure}[!t]
\centering \includegraphics[width=1.00\linewidth]{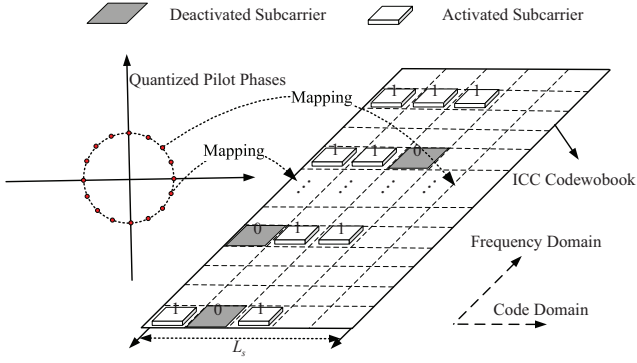}
\caption{Pilot conveying on the identified code-frequency domain.   Construct an one-to-one   mapping principle under which  the  phase candidates  in set $\cal A$ are  mapped to codewords of binary codebook matrix derived in Section~\ref{BCM}, and then further  to SAPs.  The specific principle is that pilot signals are transmitted on the $j$-th subcarrier that occupies three OFDM symbols  if  the $j$-th digit of the codeword  is equal to 1, otherwise  this subcarrier is deactivated.}
\vspace{-10pt}
\label{C_onveying}
\end{figure}

 Based on the formulated binary digits for subcarriers in detection, we denote a set of binary code vectors  by ${{\cal S}}$  with ${{\cal S}}= \left\{ {\left. {\bf{s}} \right|{s_{m}} \in \left\{ {0,1} \right\},1 \le m \le {{L_s}}} \right\}$ where ${L_s}$ denotes the maximum length of  the  code.  Then, a code frequency domain could be constructed as a set of pairs $\left( {{{\bf{s}}},b} \right)$ with ${\bf{s}} \subset {{\cal S}}$ and $1\le b \le N_{\rm B}$ where $b$ is an integer  representing the subcarrier  index of appearance of the code. This is shown in Fig.~\ref{C_onveying}.
\subsubsection{Binary Codebook Matrix}
\label{BCM}

On the formulated  code-frequency domain, we group  the binary digits and  construct  the binary code  by presenting  a binary codebook as follows:
\begin{definition}\label{E.19}
Given a $N_{\rm B} \times C$ binary matrix ${\bf C}$ with each element satisfying  ${c_{i,j}} \in {\bf{s}}\subset {{\cal S}}$, we denote the $i$-th column of ${\bf C}$  by ${\bf c}_{i}$ with ${{\bf{c}}_i} = {\left[ {\begin{array}{*{20}{c}}
{{c_{1,i}}}& \cdots &{{c_{N_{\rm B},i}}}
\end{array}} \right]^{\rm{T}}}$.   We call  $\bf C$  a binary codebook matrix  and ${{\bf{c}}_i}$ a codeword  of $\bf C$ of  length $N_{\rm B}$.
\end{definition}
 The codebook size  is equal to the quantization resolution of phases in the set $\cal A$.
 Based on this codebook matrix, a mapping from pilot phases, to codewords and further to SAPs  is developed in  Fig.~\ref{C_onveying} for pilot conveying.

Pilot conveying  provides the basis for  pilot separation and identification which also means the codeword separation and identification. Therefore, the performance of CTA becomes  totally  dependent on the property of  binary codebook.

\subsection{Pilot Separation and Identification  Via ICC}
\label{PSI}

 In this subsection, we present  the  ICC theory to  optimize   the previous binary  codebook.  Its crucial feature   is to  create  the ``difference'' by checking the independence of channels experienced by  different parties.  In what follows, we will introduce the ICC theory by formulating its encoding/decoding principle.
\subsubsection{Encoding Principle}
 \begin{figure}[!t]
\centering \includegraphics[width=1.00\linewidth]{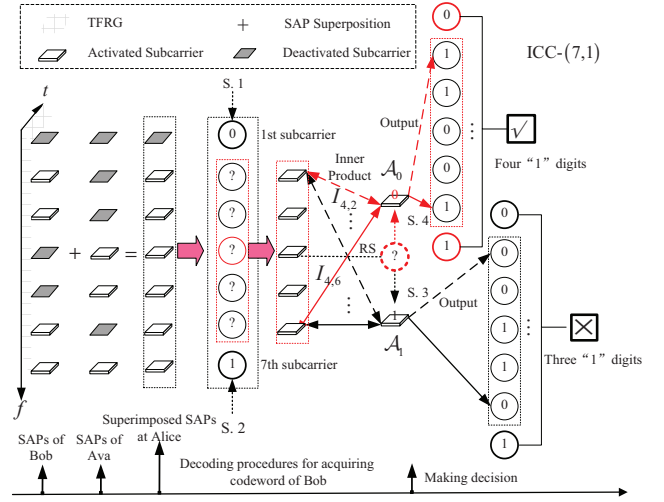}
\caption{ An example of decoding strategy using ICC-$\left( {7,1} \right)$.  After observing the superimposed SAPs from Bob and Ava  on TFRG, Alice decodes the activation patterns of subcarriers  in \textbf{Case 1}  as 0 digits and those  in \textbf{Case 2}  as 1 digits. The two operations are respectively labeled by \textbf{S.1} and  \textbf{S.2}. Alice then calculate   $N_{1}^{\rm d}$ (here  $N_{1}^{\rm d}=1$);  For \textbf{Case 3},  Alice selects  an arbitrary ambiguous subcarrier   at  position $i_{\rm r}, i_{\rm r}\in {\cal P}_{d}$ as a reference subcarrier  (RS) (here $i_{\rm r}=4$), and further  calculates the  differential code  digits ${d_{i_{\rm r},j}}$ with $ {{d_{{i_{\rm{r}}},j}} = f_{d}\left( {{I_{{i_{\rm{r}}},j}}} \right)}\bigoplus1$ where ${{I_{{i_{\rm{r}}},j}} = \left\langle {\frac{{{{\bf{y}}_{{i_{\rm{r}}}}}\left[ k \right]}}{{\left\| {{{\bf{y}}_{{i_{\rm{r}}}}}\left[ k \right]} \right\|}},\frac{{{{\bf{y}}_j}\left[ k \right]}}{{\left\| {{{\bf{y}}_j}\left[ k \right]} \right\|}}} \right\rangle }$ and $f_{d}\left( {x < r} \right) = 0,f_{d}\left( {x > r} \right) = 1,j \in {P_d}$. $\left\langle {\cdot} \right\rangle $ denotes the inner product operation and $r$ is equal to zero. When making decision, Alice  makes two assumptions, i.e.,  ${{\cal A}_1}:{d_{i_{\rm r}}} = 1$ and ${{\cal A}_0}:{d_{i_{\rm r}}} = 0$.  For each assumption ${{\cal A}_i}$, Alice  outputs  the  candidate codeword denoted by  ${\overline {\bf c}}_{i}$ by performing $\bigoplus$ operations between ${d_{i_{\rm r},j}}, j \in {\cal P}_{d}$ and ${d_{i_{\rm r}}},  i_{\rm r} \in {\cal P}_{d}$.
For example, if ${d_{{4}}}{\rm{ = 0}}$, and  ${d_{{\rm{4}},{{2}}}}{{ = 1,}}{d_{{{4}},{{3}}}}{\rm{ = 1,}}{d_{{{4}},{{5}}}}{\rm{ = 0,}}{d_{{{4}},{{6}}}}{\rm{ = 1}}$,  we can derive ${\overline {\bf{c}} _0}{\rm{ = 0110011}}$. The operations correspond to  \textbf{S.3} and  \textbf{S.4}.  Alice calculates $N_{1,i}^{\rm s}$  of ${\overline {\bf c}}_{i}$ for $i=0,1$. (For example, in this figure $N_{1,0}^{\rm s}=4$ and $N_{1,1}^{\rm s}=3$) If $N_1^{\rm d}+N_{1,i}^{\rm s}$ satisfies the weight constraint of ICC-$\left( {N_{\rm B},s} \right)$ code, the hypothesis  ${{\cal A}_i}$ is correct and Alice identifies ${\overline {\bf c}}_{i}$ as the codeword of Bob, which realizes  the codeword separation and identification.}
\vspace{-15pt}
\label{ICC}
\end{figure}

Based on the Definition 2, we further have the following definition:
\begin{definition}
A $N_{\rm B} \times C$ binary matrix $\bf C$  is called a ICC-$\left( {N_{\rm B},s} \right)$ code of length $N_{\rm B}$ and order $s$,  if for any column set $\cal Q$ such that  $ \left| \cal Q \right|=2$,  there exist at least a set $\cal S$ of $s$ rows  such that ${c_{i,j}} = 1,\forall i,j, {i \in \cal S}, j \in  {\cal Q}$.
\end{definition}
For this principle,  any two codewords in $\bf C$ must superimpose  with each other on  at least  $s$ non-zero digits.
\begin{remark}
Basically, $s, s\ge 1$ denotes the discriminatory feature we have created.  This feature intrinsically can be seen as a characteristic  that there always exist  more nonzero digits  than zero digits. Returning  to the subcarriers, $s$ means the available number of  overlapping subcarriers  for channel estimation.  The overlapping of subcarriers  means the  coexistence  of signals from two nodes on the same subcarrier and same OFDM symbol time.
\end{remark}
\begin{theorem}
The weight of ICC-$\left( {N_{\rm B},s} \right)$ code of length $N_{\rm B}$ and order $s$ satisfies $w = \frac{{N_{\rm B} + s}}{2}$ with $N_{\rm B}\ge s$. $w$ is an integer smaller than $N_{\rm B}$.
\end{theorem}
\begin{proof}
See proof in Appendix~\ref{Theorem1}
\end{proof}
Here and in the following sections, we assume the ratio of two integer is always  kept to be an integer without loss of generality.  Based on the theorem, we can derive the number of  codewords or namely the columns  in $\bf C$,  by a binomial coefficient $C = \left( {\begin{array}{*{20}{c}}
{{N_{\rm B}}}\\ {\frac{{{N_{\rm B}} + s}}{2}}
\end{array}} \right)$.
Then we have the following proposition about the code rate:
 \begin{figure*}[!t]
\centering \includegraphics[width=1.00\linewidth]{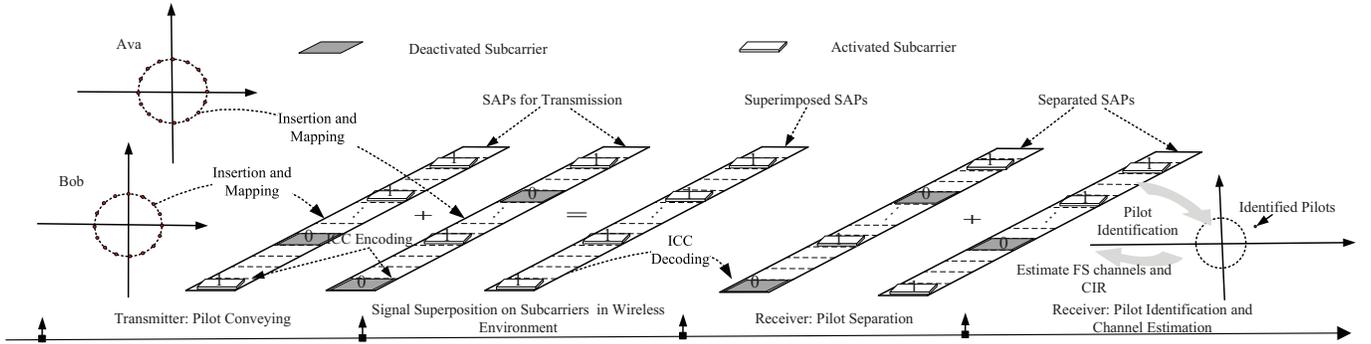}
\caption{ Diagram of   ICC-CTA protocol  procedures. Each time the training begins, Bob selects one quantized phase ${\phi _{{k_i}}}$ from set $\cal A$, for instance  $i={0}$. Bob shows the   SAPs  corresponding to ${\phi _{{k_0}}}$, as the above  mapping  principle indicates.   \textbf{Across the frequency domain},  the insertion of Bob's pilots  obeys Assumption 3. This operation applies to  all of three OFDM symbols. \textbf{For the time domain},  at the initial  symbol $k_0$, Bob inserts  onto the pilot subcarriers in activation  the pilot with phase ${\phi _{{k_0}}}$. The pilots inserted  within  adjacent OFDM symbols, such as $k_{i},i\ge1$,  obey the Assumption 4.  Those SAPs, after undergoing wireless channels, suffer from the superposition interference from each other, and finally are superimposed and observed at Alice which separates and identifies those pilots. The technical details could be seen in Fig.~\ref{ICC} and its caption below. Based  on the identified pilots, Alice performs channel estimation.}
\vspace{-10pt}
\label{Procedure}
\end{figure*}

\begin{proposition}
 The code rate of ICC-$\left( {N_{\rm B},s} \right)$ code,  defined by ${R_{ICC}} = \frac{{{{\log }_2}\left( C \right)}}{{{N_{\rm B}}}}$, is calculated as:
\begin{equation}\label{E.21}
{R_{ICC}}\left( {{N_{\rm{B}}},w} \right) = {\log _2}{\left[ {\frac{{{N_{\rm B}}!}}{{\left( {\frac{{{N_{\rm B}} + s}}{2}} \right)!\left( {\frac{{{N_{\rm B}} -s}}{2}} \right)!}}} \right]^{{1 \mathord{\left/
 {\vphantom {1 {{N_{\rm B}}}}} \right.
 \kern-\nulldelimiterspace} {{N_{\rm B}}}}}}
\end{equation}
\end{proposition}

\subsubsection{Decoding Procedure}
 Despite the fact that the encoding principle  provides  the discriminatory feature of ICC, Alice  has to construct a decoding principle according with this feature  to perform  codeword separation and identification

Considering the hybrid attack environment, Alice could recognize three types of results on  the $i$-th  subcarrier $i\in \left[ {1,{N_{\rm B}}} \right]$: \textbf{Case 1:}None of Bob and Ava  transmits signals.  \textbf{Case 2:} Bob and Ava both transmit signals.  \textbf{Case 3:} One unknown node  (Bob or Ava) transmits signals. Obviously, Alice can identify the behaviors   in the first two cases but this cannot work well in \textbf{Case 3} due to the ambiguity of  superposition operation of signals on subcarriers. For simplicity, we define the subcarriers in \textbf{Case 1} and \textbf{Case 2}  as the deterministic subcarriers while those in \textbf{Case 3} are defined as  the ambiguous subcarriers. The related decoding principle is  depicted in Fig.~\ref{ICC}.

 Now that we have explored the principle of ICC method in theory, we  ought to look at its performance evaluation.
 \begin{proposition}
 SEP, defined by  error probability of separating  two right codewords from the observed codeword,  is zero.
\end{proposition}
It is sufficiently feasible that the distance between Bob and Ava  can guarantee that their channels  fade independently with each other. The inner product of high-dimensional receiving signals on different subcarriers is therefore always precisely measured  under massive antennas,  providing the perfect differential decoding and thus perfect pilot separation in Fig.~\ref{ICC}.
\begin{theorem}
IEP, defined by the error probability of  identifying Bob's codeword from the two separated codewords,  is given by
\begin{equation}\label{E.25}
{{P}_{\rm{I}}} = \frac{{{N_{\rm B}}! - \left( {\frac{{{N_{\rm B}} + s}}{2}} \right)!\left( {\frac{{{N_{\rm B}} - s}}{2}} \right)!}}{{{2^{{N_{\rm B}} + 1}}\left( {\frac{{{N_{\rm B}} + s}}{2}} \right)!\left( {\frac{{{N_{\rm B}} - s}}{2}} \right)!}}
\end{equation}
\end{theorem}
\begin{proof}
See proof in Appendix~\ref{Theorem2}.
\end{proof}

The overall pilot conveying, separation and identification can be seen in part of Fig.~\ref{Procedure}.
\section{FS Channel Estimation and Security Enhancement}
\label{FSCE}
In this section,  we continue  our design work  for the ICC-CTA protocol architecture and focus on the FS channel estimation.  Two  questions  will be answered further:
\begin{question}
How to estimate FS channels based on the identified pilots?
\end{question}
\begin{question}
 Is it possible to  improve the security  performance of ICC theory by further digging the properties of estimated FS channels ?
\end{question}
\subsection{FS Channel Estimation}
It is well-known that  LS estimator is a natural choice when there is no attack.  In this subsection, we only consider the FS channel estimation under PTJ attack shown in the attack model  in Introduction part.

 In principle, performing linear channel estimation  requires specifying  the receiving signal model and  linear decorrelating estimator (LDE) that weights on the receiving signals for channel estimation.

 Let us consider  the construction of LDE. Basically, Alice  examines the decoded pilots which  can be,  1) successfully identified; ( no identification error) or 2)  confusing  (identification error happens).  We in this section consider the latter and forget the case without identification error. In this way,  the  estimator to be designed naturally apply to the case without identification error.  Within two OFDM symbol time,  i.e., ${k_0}$ and ${k_1}$,  Alice could collect  two confusing pilot vectors defined by  ${\bf{x}}_{{\rm{L,1}}}$ and ${\bf{x}}_{{\rm{L,2}}}$ where ${\bf{x}}_{{\rm{L,1}}}{\rm{ = }}{\left[ {\begin{array}{*{20}{c}}
{{x_{\rm{B}}}\left[ {{k_0}} \right]}&{{x_{\rm{B}}}\left[ {{k_1}} \right]}
\end{array}} \right]^{\rm{T}}}$ and ${\bf{x}}_{{\rm{L,2}}}{\rm{ = }}{\left[ {\begin{array}{*{20}{c}}
{{x_{\rm{A}}}\left[ {{k_0}} \right]}&{{x_{\rm{A}}}\left[ {{k_1}} \right]}
\end{array}} \right]^{\rm{T}}}$. The notation of ${{x_{\rm{B}}}\left[ {{k}} \right]}$ can be found in Assumption 3. Here the confusing case happens  when Ava keeps the same frequency-domain and time-domain PIP as Bob, which is proved in Remark 2. Then we use the notation of ${{x_{\rm{A}}}\left[ {{k}} \right]}$  with the only difference, that is, different value with ${{x_{\rm{B}}}\left[ {{k}} \right]}$.

Then we consider  the receiving  signal model  for which two  facts involved  should be clarified:
\begin{fact}1) The phenomenon that arbitrary  two codewords within ICC-$\left( {N_{\rm B},s} \right)$  must overlap  at least on $s$ code digits  does not mean that the total number of overlapping subcarriers always keeps  stable and constant;
2) The superimposed signals on those overlapping subcarriers  could  be employed for channel estimation and security enhancement whereas the subcarrier on which  only one signal exists  can  be utilized for, but limited to channel estimation.
\end{fact}
In order  to formulate the receiving signal, we choose two OFDM symbol time, i.e., ${k_0}$ and ${k_1}$, and $s,s\ge1 $ randomly-overlapping subcarriers. The randomness here  means the  random frequency positions of subcarriers.  The signals received  are  stacked  as the ${2 \times  {N_{\rm T}s}}$ matrix ${{\bf{Y}}_{\rm L}}$, equal to
\begin{equation}\label{E.26}
{{\bf{Y}}_{\rm L}} = {{\bf{X}}_{{\rm L}}}{{\bf{H}}_{{{{\rm L}}}}} + {{\bf{N}}_{\rm L}}
 \end{equation}
 where the ${2 \times 2}$ matrix  ${{\bf{X}}_{{\rm L}}}$ satisfies  ${{\bf{X}}_{\rm L}} = \left[ {\begin{array}{*{20}{c}}
{\bf{x}}_{{\rm{L,1}}}&{\bf{x}}_{{\rm{L,2}}}
\end{array}} \right]$.  The integrated ${2 \times  {N_{\rm T}s}}$  channel matrix ${{\bf{H}}_{{{{\rm L}}}}}$  satisfies    ${{\bf{H}}_{{{{\rm L}}}}} = {\left[ {\begin{array}{*{20}{c}}
{{\bf{h}} _{{\rm{B}},{\rm L}}^{\rm{T}}}&{{\bf{h}} _{{\rm{A}},{\rm L}}^{\rm{T}}}
\end{array}} \right]^{\rm{T}}}$ where  ${{\bf{h}}_{{\rm{B}},{\rm L}}} = \left[ {\begin{array}{*{20}{c}}
{{{\left( {{{\bf{F}}_{{\rm{L}}, s}}{\bf{h}}_{\rm{B}}^1} \right)}^{\rm{T}}}}&{, \ldots ,}&{{{\left( {{{\bf{F}}_{{\rm{L}}, s}}{\bf{h}}_{\rm{B}}^{{N_{\rm{T}}}}} \right)}^{\rm{T}}}}
\end{array}} \right]$ and ${{\bf{h}}_{{\rm{A}},{\rm L}}} = \left[ {\begin{array}{*{20}{c}}
{{{\left( {{{\bf{F}}_{{\rm{L}}, s}}{\bf{h}}_{\rm{A}}^1} \right)}^{\rm{T}}}}&{, \ldots ,}&{{{\left( {{{\bf{F}}_{{\rm{L}}, s}}{\bf{h}}_{\rm{A}}^{{N_{\rm{T}}}}} \right)}^{\rm{T}}}}
\end{array}} \right]$. ${{\bf{F}}_{{\rm{L}}, s}}$ is  the $s$-row matrix for which each index of $s$ rows belongs to the set ${\cal P}_{s}$.
${\bf{N}}_{\rm L}$ represents the ${2 \times  {N_{\rm T}s}}$ noise matrix  with  ${{\bf{N}}_{\rm L}} = {\left[ {\begin{array}{*{20}{c}}
{{\bf{w}}_{\rm L}^{\rm{T}}\left[ {{k_0}} \right]}&{{\bf{w}}_{\rm L}^{\rm{T}}\left[ {{k_1}} \right]}
\end{array}} \right]^{\rm{T}}}$ where ${{\bf{w}}_{\rm L}}\left[ k \right] = \left[ {\begin{array}{*{20}{c}}
{{{\bf{w}}_{s}^{{1^{\rm{T}}}}}\left[ k \right]}&{, \ldots ,}&{{{\bf{w}}_{s}^{{N_{\rm{T}}}^{\rm{T}}}}\left[ k \right]}
\end{array}} \right]$ for $k=k_{0}, k_{1}$.
\begin{remark}
Since the specific values of elements in ${\cal P}_{s}$ are randomly distributed between 1 and $N$, the ${{\bf{F}}_{{\rm{L}}, s}}$ is no longer  a semi-unitary matrix.
\end{remark}
 We formulate the sample covariance matrix   by ${{\bf{C}}_{{{\bf{Y}}_{\rm L}}}} = \frac{1}{{{N_{\rm{T}}}s}}{{\bf{Y}}_{\rm L}}{\bf{Y}}_{\rm L}^{\rm{H}}$ and  then could derive  the  asymptotically-optimal linear minimum mean square error (LMMSE) estimators as ${{\bf{W}}_{{\rm{B}},{\rm{L}}}} = {T_{\rm{B}}} {\bf{x}}_{{\rm{L,1}}}^{\rm{H}}{\bf{C}}_{{{\bf{Y}}_{\rm{L}}}}^{ - 1}$ and ${{\bf{W}}_{{\rm{A}},{\rm{L}}}} = {T_{\rm{A}}} {\bf{x}}_{{\rm{L,2}}}^{\rm{H}}{\bf{C}}_{{{\bf{Y}}_{\rm{L}}}}^{ - 1},$
  where ${T_{\rm{B}}} \buildrel \Delta \over = \frac{{{\rm{Tr}}\left( {{{\bf{R}}_{{1}}}} \right){\rm{Tr}}\left( {{{\bf{R}}_{\rm{F}}}} \right)}}{{{N_{\rm{T}}}s}}$ and ${T_{\rm{A}}} \buildrel \Delta \over = \frac{{{\rm{Tr}}\left( {{{\bf{R}}_{2}}} \right){\rm{Tr}}\left( {{{\bf{R}}_{\rm{F}}}} \right)}}{{{N_{\rm{T}}}s}}$. Here,  there exists ${\rm{Tr}}\left( {{{\bf{R}}_{{1}}}} \right) = {\rm{Tr}}\left( {{{\bf{R}}_{{2}}}} \right)=N_{\rm T}$ and therefore we could define ${T_{\rm{B}}} = {T_{\rm{A}}} = T$.

  Finally, the estimated versions of FS channels are  respectively  derived  as
  \begin{equation}\label{E.27.5}
   {\widehat {\bf{h}}_{{\rm{B}},{\rm{L}}}}{\rm{ = }}{{\bf{W}}_{{\rm{B}},{\rm{L}}}}{{\bf{Y}}_{\rm{L}}},  {\widehat {\bf{h}}_{{\rm{A}},{\rm{L}}}}{\rm{ = }}{{\bf{W}}_{{\rm{A}},{\rm{L}}}}{{\bf{Y}}_{\rm{L}}}
   \end{equation}
  The normalized mean square error (NMSE)  for  the two estimations are respectively defined by
$\varepsilon _{\rm{B}}^2 = \frac{{{\mathbb E}\left\{ {{{\left\| {{{\widehat {\bf{h}}}_{\rm{B,L}}} - {{\bf{h}}_{\rm{B,L}}}} \right\|}^2}} \right\}}}{{{N_{\rm{T}}}s}},\varepsilon _{\rm{A}}^2 = \frac{{{\mathbb E}\left\{ {{{\left\| {{{\widehat {\bf{h}}}_{\rm{A,L}}} - {{\bf{h}}_{\rm{A,L}}}} \right\|}^2}} \right\}}}{{{N_{\rm{T}}}s}}$.
Furthermore, the relationship between the ideal channels  with estimated versions can be given by
 ${{\bf{h}}_{\rm{B,L}}} = {\widehat {\bf{h}}_{\rm{B,L}}} + {\varepsilon _{\rm{B}}}{\bf{h}}$ and ${{\bf{h}}_{\rm{A,L}}} = {\widehat {\bf{h}}_{\rm{A,L}}} + {\varepsilon _{\rm{A}}}{\bf{h}}^{'}$
 where $ {\varepsilon _{\rm{B}}}{\bf{h}}$  is uncorrelated with ${{\bf{h}}_{\rm{B,L}}}$ and $ {\varepsilon _{\rm{A}}}{\bf{h}}^{'}$  is uncorrelated with ${{\bf{h}}_{\rm{A,L}}}$. Here, the entries of ${\bf{h}}$  and ${\bf{h}}^{'}$ are i.i.d zero-mean complex Gaussian vectors with each  element having unity variance.
  \begin{proposition}
  In the large-scale  array regime, there exists $\varepsilon _{\rm{B}}^2 = \varepsilon _{\rm{A}}^2$ at  high SNR .
    \end{proposition}
 \begin{IEEEproof}
See proof in Appendix~\ref{Proposition3}
 \end{IEEEproof}
 \begin{remark}
   When no identification error happens, Alice only utilizes the identified pilots of Bob to derive  ${\bf{x}}_{{\rm{L,1}}}$ and finally gets ${\widehat {\bf{h}}_{{\rm{B}},{\rm{L}}}}$.
  \end{remark}

\subsection{Security Enhancement: Exploiting Spatial Correlation}
 We are now ready   to answer Question 2.  Security enhancement in this section means reducing IEP further. To this end,  we should focus on the case where Bob gets two confusing pilots, i.e, ${\bf{x}}_{{\rm{L,1}}}$ and ${\bf{x}}_{{\rm{L,2}}}$ and two confusing estimated channels, i.e., ${\widehat {\bf{h}}_{{\rm{B}},{\rm{L}}}}$ and ${\widehat {\bf{h}}_{{\rm{A}},{\rm{L}}}}$. Even in this case, the identification  error will occur only when  Ava keeps the same frequency-domain and  time-domain PIP as Bob, which is proved in Remark 2. In this section, we will reduce the probability of this happening  in an independent dimension, i.e., the angular domain.
\subsubsection{Angular Domain Identification}
Basically, the process of identification can be modelled  as a decision process between two hypotheses:
 \begin{equation}\label{E.33}
 \begin{array}{l}
{{\cal H}_{\rm{0}}}:{\widehat {\bf{h}}_{{\rm{B}},{\rm{L}}}} \to Bob, ~
{\cal {H}_{\rm{1}}}:{\widehat {\bf{h}}_{{\rm{A}},{\rm{L}}}} \to Bob
\end{array}
\end{equation}
 For the sake of simplicity,  we define several useful  eigenvalue decompositions, including
 $
{{\bf{R}}_i} = {{\bf{U}}_i}{{\bf{\Lambda }}_i}{\bf{U}}_i^{\rm{H}}$,
${\overline {\bf{R}} _i} = {{\bf{U}}_i}{\overline {\bf{\Lambda }} _i}{\bf{U}}_i^{\rm{H}}$,  ${{\bf{R}}_{\rm{F}}} = {{\bf{V}}_{\rm{f}}}{{\bf{\Sigma }}_{\rm{f}}}{\bf{V}}_{\rm{f}}^{\rm{H}}$ and  ${\overline {\bf{R}} _{\rm{F}}} = {{\bf{V}}_{\rm{f}}}{\overline {\bf{\Sigma }} _{\rm{f}}}{\bf{V}}_{\rm{f}}^{\rm{H}}$. Here, ${\bf{U}}_i$ and ${{\bf{V}}_{\rm{f}}}$ denote the eigenvector matrices and  eigenvalue matrices satisfy ${{\bf{\Lambda }}_i} = {\rm{diag}}\left\{ {{{\left[ {\begin{array}{*{20}{c}}
{{\lambda _{i,1}}}& \cdots &{{\lambda _{i,{\rho _i}}}}&0& \cdots &0
\end{array}} \right]}^{\rm{T}}}} \right\}$, ${\overline {\bf{\Lambda }} _i} = {\rm{diag}}\left\{ {{{\left[ {\begin{array}{*{20}{c}}
{\lambda _{i,1}^{ - 1}}& \cdots &{\lambda _{i,{\rho _i}}^{ - 1}}&0& \cdots &0
\end{array}} \right]}^{\rm{T}}}} \right\}$, ${{\bf{\Sigma }}_{\rm{f}}} = {\rm{diag}}\left\{ {{{\left[ {\begin{array}{*{20}{c}}
{{\lambda _{{\rm{f}},1}}}& \cdots &{{\lambda _{{\rm{f}},{\rho _{\rm{f}}}}}}&0& \cdots &0
\end{array}} \right]}^{\rm{T}}}} \right\}$,  ${\overline {\bf{\Sigma }} _{\rm{f}}} = {\rm{diag}}\left\{ {{{\left[ {\begin{array}{*{20}{c}}
{\lambda _{{\rm{f}},1}^{ - 1}}& \cdots &{\lambda _{{\rm{f}},{\rho _{\rm{f}}}}^{ - 1}}&0& \cdots &0
\end{array}} \right]}^{\rm{T}}}} \right\}$.

We build up an error decision function as
     \begin{equation}\label{E.34}
\Delta f  \buildrel \Delta \over = f\left( {{{\widehat {\bf{h}}}_{{\rm{B}},{\rm{L}}}}} \right) - f\left( {{{\widehat {\bf{h}}}_{{\rm{A}},{\rm{L}}}}} \right)
  \end{equation}
  where   $f\left( {\bf{r}} \right) = {\bf{r}}\left( {{\overline {\bf{R}} _1} \otimes {\overline {\bf{R}} _{\rm{F}}}} \right){{\bf{r}}^{\rm{H}}}$.   Then we have the following theorem to identify two hypotheses.
\begin{theorem}
When ${N_{\rm{T}}} \to \infty $, the error decision function can be simplified as:
 \begin{equation}\label{E.35}
  \Delta f  = L\left\{ {{\rho _{\rm{1}}} - {\rm{Tr}}\left( {{{\bf{R}}_2}{{\overline {\bf{R}} }_1}} \right)} \right\}
  \end{equation}
 \end{theorem}
 \begin{IEEEproof}
See proof in Appendix~\ref{Theorem3}
\end{IEEEproof}

The further simplification of above equation requires exploiting the relationship between ${{\bf{R}}_1}$ and ${{\bf{R}}_2}$.
Backing to the Eq.~(\ref{E.35}), we know that  the trace function satisfies  ${\rm{Tr}}\left( {{{\bf{R}}_2}{{\overline {\bf{R}} }_1}} \right) \le {\rm{Tr}}\left( {{{\bf{\Lambda }}_2}{\bf{U}}_2^{\rm{H}}{{\bf{U}}_1}{{\overline {\bf{\Lambda }} }_1}} \right)={\rm{Tr}}\left( {{{\bf{\Lambda }}_{2,{\rm{p}}}}\overline {\bf{U}} _2^{\rm{H}}{{\overline {\bf{U}} }_1}{{\overline {\bf{\Lambda }} }_{1,{\rm{p}}}}} \right)$ where  ${{\bf{\Lambda }}_{i,{\rm{p}}}} $ and ${\overline {\bf{\Lambda }} _{i,{\rm{p}}}}$ are respectively defined by ${{\bf{\Lambda }}_{i,{\rm{p}}}} = {\rm{diag}}\left\{ {{{\left[ {\begin{array}{*{20}{c}}
{{\lambda _{i,1}}}& \cdots &{{\lambda _{i,{\rho _i}}}}
\end{array}} \right]}^{\rm{T}}}} \right\}$ and ${\overline {\bf{\Lambda }} _{i,{\rm{p}}}} = {\rm{diag}}\left\{ {{{\left[ {\begin{array}{*{20}{c}}
{\lambda _{i,1}^{ - 1}}& \cdots &{\lambda _{i,{\rho _i}}^{ - 1}}
\end{array}} \right]}^{\rm{T}}}} \right\}$. The $N_{{\rm T}} \times {\rho _{{i}}}$ matrix ${{{\overline {\bf{U}} }_i}}$ denotes  the tall unitary matrix of  channel covariance eigenvectors ${{\bf{U}}_i}$.
As  discussed in~\cite{Adhikary},  ${\overline {\bf{U}} _2^{\rm{H}}{{\overline {\bf{U}} }_1}}$ can be approximated using  ${\bf{F}}_{{{\cal S}_2}}^{\rm{H}}{{\bf{F}}_{{{\cal S}_1}}}$. We define  ${{\cal S}_1} \cap {{\cal S}_2} = {{\cal S}_3}$ where ${{\cal S}_i}$ denotes the support of ${S}_{i}\left( x \right)$,   a uniformly-bounded absolutely-integrable function satisfying  ${S_i}\left( x \right) = \frac{1}{{2\Delta }}\sum\limits_{0 \in \left[ {D\sin \left( {{\theta _i} - \Delta } \right) + x,D\sin \left( {{\theta _i} + \Delta } \right) + x} \right]} {\frac{1}{{\sqrt {{D^2} - {x^2}} }}}$,  over $x \in \left[ { - \frac{1}{2},\frac{1}{2}} \right]$. There exists ${{\bf{F}}_{{{\cal S}_i}}} = \left( {{{\bf{f}}_n}:n \in {{\cal J}_{{{\cal S}_i}}}} \right)$ where
${{\cal J}_{{{\cal S}_i}}} = \left\{ {n,\left[ {{n \mathord{\left/
 {\vphantom {n {{N_{\rm{T}}}}}} \right.
 \kern-\nulldelimiterspace} {{N_{\rm{T}}}}}} \right] \in {{\cal S}_i},n = 0, \ldots ,{N_{\rm{T}}} - 1} \right\}$.
We  then discuss the influence of  ${{\cal S}_3}$ on ${\rm{Tr}}\left( {{{\bf{R}}_2}{{\overline {\bf{R}} }_1}} \right)$.  When ${{\cal S}_3} = \emptyset $, we can have ${\rm{Tr}}\left( {{{\bf{R}}_2}{{\overline {\bf{R}} }_1}} \right)=0$. When ${{\cal S}_3}\ne \emptyset$, we assume ${{\cal S}_3} = {\cal P}_{a}$ and have
  \begin{equation}\label{E.39}
{\rm{Tr}}\left( {{{\bf{\Lambda }}_{2,{\rm{p}}}}\overline {\bf{U}} _2^{\rm{H}}{{\overline {\bf{U}} }_1}{{\overline {\bf{\Lambda }} }_{1,{\rm{p}}}}} \right) \le \sum\limits_{j = 1}^a {\frac{{{\lambda _{2,{i_j}}}}}{{{\lambda _{1,{i_j}}}}}}
  \end{equation}
 This is because  the eigenvectors labeled by the indexes out of  the interacted set ${{\cal S}_3}$ are mutually orthogonal~\cite{Adhikary}.
  \begin{theorem}
 When ${N_{\rm{T}}} \to \infty $,  there always exists $\sum\limits_{j = 1}^a {\frac{{{\lambda _{2,{i_j}}}}}{{{\lambda _{1,{i_j}}}}}} =  a$.  If ${\theta _1} \ne {\theta _2}$,   there must exist  $a< {\rho _1} $ and $ \Delta f > 0$. Otherwise if $ {\theta _1} = {\theta _2}$, there must exist  $a={\rho _1}$ and $ \Delta f =0$.
\end{theorem}
\begin{IEEEproof}
See proof in Appendix~\ref{Theorem4}.
 \end{IEEEproof}
 Thus far, we can know that  Ava is restricted on  a line lying the center of  clusters surrounding  Bob,  otherwise, its attack is invalidated, which  shows another potential of angular domain identification in countering attack.
 \begin{algorithm}[!t]\label{CEAIE}
\caption{:Channel  Estimation and  Security  Enhancement}
\begin{algorithmic}[1]
{
\STATE  Identify whether or not  PTJ attack happens using  the codewords decoded as shown in Fig. 5.
\IF {PTJ attack happens}
\STATE Derive ${\bf{x}}_{{\rm{L,1}}}$, ${\bf{x}}_{{\rm{L,2}}}$,  and  ${\widehat {\bf{h}}_{{\rm{B}},{\rm{L}}}}$, ${\widehat {\bf{h}}_{{\rm{A}},{\rm{L}}}}$  using Eq.~(\ref{E.27.5}).
\ELSE
\STATE   Apply LS estimator to derive ${\widehat {\bf{h}}_{{\rm{B}},{\rm{L}}}}$.
\ENDIF
\IF {No PTJ attack  happens}
\STATE Directly derive CIR estimation for Bob using ${\widehat {\bf{h}}_{{\rm{B}},{\rm{L}}}}$.
\ELSE
\STATE Detect if the confusing case occurs.
\IF {No confusing case happens}
\STATE  Use ${\widehat {\bf{h}}_{{\rm{B}},{\rm{L}}}}$  for CIR estimation.
\ELSE
\STATE Calculate $\Delta f $ using Eq.~(\ref{E.34}) and Theorem 4.
\IF {$\Delta f>0 $}
\STATE ${\widehat {\bf{h}}_{{\rm{B}},{\rm{L}}}}$ is used  for Bob's CIR estimation.
\ELSIF {$\Delta f<0 $}
\STATE  ${\widehat {\bf{h}}_{{\rm{A}},{\rm{L}}}}$  is used  for Bob's CIR estimation.
\ELSIF{$\Delta f=0 $}
\STATE Identification error happens.
\ENDIF
\ENDIF
\ENDIF
}
\end{algorithmic}
\end{algorithm}

 \subsubsection{ Combine Angular  Domain with Code Domain to Enhance Security }
 Since the pilot identification breaks down  iff  ${\theta _1} = {\theta _2}$, we have the following theorem:
  \begin{theorem}
 Under the assumption of  mean AoA obeying  CPD, the IEP ${{P}_{\rm{I}}} $ is equal to zero.
 Under the assumption of  mean AoA obeying  DPD, for instance, uniform distribution with interval length $K$,   the IEP ${{P}_{\rm{I}}} $ is updated  to be $\frac{{{P_{\rm{I}}}}}{K}$.
  \end{theorem}
 The proof is institutive since we consider two independent  dimensions, that is, angular domain and code domain,   to  reduce IEP. The IEP is lowered to  ${{P}_{\rm{I}}} $  by using coding approach and further reduced to $\frac{{{P_{\rm{I}}}}}{K}$ by exploiting  angular  domain identification. In this sense, the security provided on the code domain by the  ICC-CTA protocol is enhanced at the same time by fully exploiting the angular domain. Finally, we give the overall process of  channel estimation and  security enhancement  in Algorithm 4.
\begin{remark}
 We  aim to evaluate the influence  of different PIP principles of Ava on Theorem 4.  We need to stress that the key  lies in the following two aspects.
 On one hand, Ava  selects  different frequency-domain PIP principles with Bob. It  adopts different phases across its own activated  subcarriers in order to protect its own correlation property from being exploited  by Alice.  In this case, the original DFT submatrix in ${\bf H}_{\rm L}$ of  Eq.~(\ref{E.26}) is now represented by ${\widetilde {\bf{F}}_{{\rm{L}},{{s}}}}$ with ${\widetilde {\bf{F}}_{{\rm{L}},{{s}}}} = {\bf{\Psi }}{{\bf{F}}_{{\rm{L}},{{s}}}}$. Here, ${\bf{\Psi }} = diag\left\{ {{{\left[ {\begin{array}{*{20}{c}}
{{e^{j{\beta _1}}}}& \cdots &{{e^{j{\beta _s}}}}
\end{array}} \right]}^{\rm{T}}}} \right\}$ represents the strategies of Ava across subcarriers on which  ${\beta _i},i = 1, \ldots ,{{{s}}}$ are random.  As we can see,  there exists  ${\widetilde {\bf{R}}_{\rm{F}}} = \widetilde {\bf{F}}_{{\rm{L}},{{s}}}^{\rm{T}}\widetilde {\bf{F}}_{{\rm{L}},{{s}}}^* = {{\bf{R}}_{\rm{F}}}$. This does not affect  the value of function $f$ and thus not violate the Theorem 4.
On the other hand, we examine the case where Ava adopts different time-domain PIP principles with Bob.  In this case, LLE vector derived by Bob  is not optimal for Ava's channel estimation since  the final pilot vectors demapped from Ava's SAPs are actually wrong for channel estimation.  The elements of  ${\widehat {\bf{h}}_{{\rm{A}},{\rm{L}}}}$ in Eq.~(\ref{E.27.5}) are  further imposed on significant estimation error.  Thus Bob acquires  very large $\varepsilon _{\rm{A}}^2 $,  compared with  $\varepsilon _{\rm{B}}^2 $ derived under  asymptotically-optimal LMMSE estimation. Finally, the value of $\Delta f$ must be much  larger than zero, which does not violate the Theorem 4. Actually, this can guarantee perfect security even  $ {\theta _1} = {\theta _2}$.

In summary,  those PIP principles different with Bob's strategy can benefit Alice and  are not prudent for Ava.
\end{remark}

\section{ Security-Instability Tradeoff  in CIR Estimation}
\label{RST}

Security advantages originate from  the diversified SAPs using  ICC-$\left( {N_{\rm B},s} \right)$.  However, various  superimposed modes of SAPs (SSAPs)  affect the stability of  CIR  estimation significantly as  those subcarriers in activation are utilized for estimating CIR samples from estimated FS channels.  To begin with, we   show when and why this instability could occur and then gradually wean ourselves from the constraint of instability to find a tradeoff between the security  and instability in CIR estimations.  Finally, we  present  an optimal code rate under which a sufficiently-stable  estimation performance is  secured.

\subsection{Essence of  Unstable CIR Estimation: Random SSAPs}
Recall that each pilot phase in use has been mapped to one  unique SAP and thus randomized pilots mean  random SSAPs.  When random SAPs from Bob and Ava are superimposed in wireless environment, Alice will observe two typical SSAPs which both incur unstable performance. This can be seen in Fig.~\ref{PSOM11}. The key question is:
 \emph{How to evaluate  and reduce the  influence of the instability  resulting from  random SSAPs on CIR estimation ?}

 To answer this question, let us focus on the mathematical expression  of CIR estimation.  The CIR generally  satisfies  the equation  ${{\bf{h}}_{{\rm{B}},{\rm{L}}}} \buildrel \Delta \over = {{\bf{g}}_{{\rm{B}},{\rm{L}}}}\left( {{\bf{R}}_{\rm{1}}^{{1 \mathord{\left/
 {\vphantom {1 2}} \right.
 \kern-\nulldelimiterspace} 2}} \otimes  {\bf{F}}_{{\rm{L}},{{s}}}^{\rm{T}}} \right)$ where ${{\bf{g}}_{{\rm{B}},{\rm{L}}}} $ is the integrated $1 \times N_{\rm T}L$  CIR vector  of i.i.d.  ${\cal C}{\cal N}\left( {0,{1}} \right)$ random variables. Given ${\bf{R}}_{1}$ and ${{\bf{h}}_{{\rm{B}},{\rm{L}}}}$, the estimation  of ${{\bf{g}}_{{\rm{B}},{\rm{L}}}} $, denoted by ${\widehat {\bf{g}}_{{\rm{B}},{\rm{L}}}}$ , will  fluctuate under various forms  of ${\bf{F}}_{{\rm{L}},{{s}}}^{\rm{T}}$. Note that the structure of ${\bf{F}}_{{\rm{L}},{{s}}}^{\rm{T}}$  is determined by  the  number $s$   and the frequency positions  of  overlapping subcarriers.  Therefore, the key factor  influencing the stability of  ${\widehat {\bf{g}}_{{\rm{B}},{\rm{L}}}}$  is  ${\cal P}_{s}$.

\begin{figure}[!t]
\centering \includegraphics[width=1.00\linewidth]{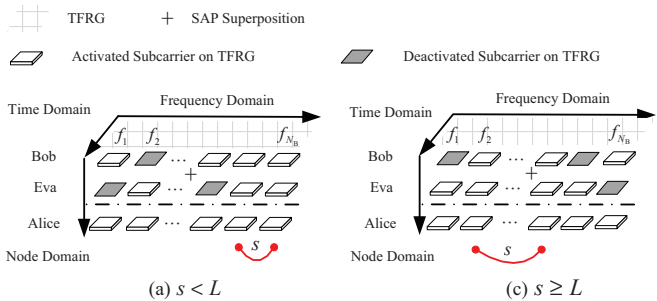}
\caption{ Diagram of random  SSAPs observed at Alice on TFRG: (a) $s$ is less than ${ L}$, thus also less than ${N_{\rm B}}$.  CIR estimation solely relies  on those  non-overlapping subcarriers that  belong to Bob. This operation is  unstable and the requirement for subcarrier identities is hard to guarantee because  the occurrence of  subcarriers without being affected by Ava is random. (b) $s$ is  larger  than ${ L}$. CIR estimation  relies  on the overlapping subcarriers. Nevertheless, the distribution of those subcarriers is also unstable, causing unstable CIR estimation.}
\label{PSOM11}
\vspace{-10pt}
\end{figure}
Specifically, we examine  Fig.~\ref{PSOM2}~(a).  When $s<L$,  the CIR estimation from ${{\bf{h}}_{{\rm{B}},{\rm{L}}}}$ is under-determined with low estimation precision. We turn to consider  $s\ge L$ in Fig.~\ref{PSOM2}~(b)  where we could always  find a non-underdetermined recovery model.

Nevertheless, the fluctuation of $s$ will directly influence the estimation stability.  Particularly,  the random value of  those elements in ${\cal P}_{s}$ will cause unequally-spaced overlapping subcarriers which continue to cause  instability and limited estimation precision.
To show this mathematically, we  begin by giving  the CIR estimation  ${\widehat {\bf{g}}_{{\rm{B}},{\rm{L}}}}$ as
${\widehat {\bf{g}}_{{\rm{B}},{\rm{L}}}} = {\widehat {\bf{h}}_{{\rm{B}},{\rm{L}}}}\left\{ {\overline {\bf{R}} _{\rm{1}}^{{{1} \mathord{\left/
 {\vphantom {{ 1} 2}} \right.
 \kern-\nulldelimiterspace} 2}} \otimes \left( {{\bf{F}}_{{\rm{L}},{{s}}}^{\rm{*}}{\bf{R}}_{\rm{F}}^{ - 1}} \right)} \right\}$.
 By using ${{\bf{h}}_{\rm{B,L}}} = {\widehat {\bf{h}}_{\rm{B,L}}} + {\varepsilon _{\rm{B}}}{\bf{h}}$, we then expand the equation   into
${\widehat {\bf{g}}_{{\rm{B}},{\rm{L}}}} = L{{\bf{g}}_{{\rm{B}},{\rm{L}}}}\left\{ {\left( {{\bf{R}}_{\rm{1}}^{{1 \mathord{\left/
 {\vphantom {1 2}} \right.
 \kern-\nulldelimiterspace} 2}}\overline {\bf{R}} _{\rm{1}}^{{{ 1} \mathord{\left/
 {\vphantom {{ 1} 2}} \right.
 \kern-\nulldelimiterspace} 2}}} \right)} \right\} - { {\varepsilon _{\rm{B}}}}{\bf{h}}\left\{ {\overline {\bf{R}} _{\rm{1}}^{{{ 1} \mathord{\left/
 {\vphantom {{ 1} 2}} \right.
 \kern-\nulldelimiterspace} 2}} \otimes \left( {{\bf{F}}_{{\rm{L}},{{s}}}^{\rm{*}}{\bf{R}}_{\rm{F}}^{ - 1}} \right)} \right\}$.
  Given the correlation matrix $\bf R_{1}$ and ${ {\varepsilon _{\rm{B}}}}$, the minimization of  $\overline \varepsilon  _{\rm{B}}^2$ defined by the equation  $\overline \varepsilon  _{\rm{B}}^2 = {{{\mathbb E}\left\{ {{{\left\| {{{\widehat {\bf{g}}}_{{\rm{B}},{\rm{L}}}} - {{\bf{g}}_{{\rm{B}},{\rm{L}}}}} \right\|}^2}} \right\}} \mathord{\left/
 {\vphantom {{E\left\{ {{{\left\| {{{\widehat {\bf{g}}}_{{\rm{B}},{\rm{L}}}} - {{\bf{g}}_{{\rm{B}},{\rm{L}}}}} \right\|}^2}} \right\}} {{N_{\rm{T}}}L}}} \right.
 \kern-\nulldelimiterspace} {{N_{\rm{T}}}L}}$,  is equivalent to:
  \begin{equation}\label{E.42}
 \begin{aligned}
 & \text{\,} \underset{{{{\bf{R}}_{\rm{F}}}}}{\text{min}}
 & &  {\rm{Tr}}\left( {{\bf{R}}_{\rm{F}}^{ - 1}} \right),  ~{\rm{s. t.}}~~ {\rm{Tr}}\left( {{{\bf{R}}_{\rm{F}}}} \right) = L \\
 \end{aligned}
 \end{equation}
 For this optimization problem,  the minimization is  achieved  iff ${{{\bf{R}}_{\rm{F}}}}$ has the identical eigenvalues, and thus  the overlapping subcarriers are equally spaced, satisfying
 \begin{equation}\label{E.43}
   {\overline{ \cal P}}_{s}: \left\{ {{i_k},{i_{k + \frac{N}{L}}} \ldots ,{i_{k + \frac{{\left( {L - 1} \right)N}}{L}}},k = 0,1, \ldots ,\frac{N}{L} - 1} \right\}
 \end{equation}
 The total number of subcarriers within the interval that  extends from the first  overlapping position  to the last one  can be derived as:
 \begin{equation}\label{E.44}
 s^{*} \buildrel \Delta \over =\frac{{L - 1}}{L}N + 1
  \end{equation}
Hinted by this, we know how any mismatch  between the  indices of ${{ \cal P}}_{s}$ with those of   ${\overline{ \cal P}}_{s}$  could  increase the estimation error and  instability.

Based on above observations, we define  the condition of being stable (\textbf{CS}) for CIR estimation as follows:
 \begin{definition}[\textbf{CS}]
The overlapping subcarriers   are equally spaced  and meet  the number constraint, that is ,  $s\ge L$ and $s\ge  s^{*}$
 \end{definition}

  Returning to examine the previous  SSAPs in Fig.~\ref{PSOM11}, we can know that  SAPs are diversified, completely  under the direction of  ICC-$\left( {N_{\rm B},s} \right), s \ge 1$  code.   Basically, the instability originates from the random use of codewords and the constraint of $N_{\rm B}$ and $w$ in  ICC-$\left( {N_{\rm B},s} \right), s \ge 1$  code. Therefore. any mechanism for reduction of instability must reconsider the code design. In this design process,  we must deal with the relationship between security and instability.
\subsection{Security-Instability Tradeoff }
To begin with,  we  identify and define the  instability by the following metirc:
\begin{definition}
The KPI indicating  the instability of  CIR estimation using  ICC-$\left( {N_{\rm B},s} \right)$ code  is defined by  $S_{T}\left( {{N_{\rm{B}}},w,{s^*}} \right) = {1 \mathord{\left/
 {\vphantom {1 {{P_{\rm{s}}}\left( {{N_{\rm{B}}},w,{s^*}} \right)}}} \right.
 \kern-\nulldelimiterspace} {{P_{\rm{s}}}\left( {{N_{\rm{B}}},w,{s^*}} \right)}}$ with
 \begin{equation}\label{E.45}
{P_{\rm{s}}}\left( {{N_{\rm{B}}},w,s^{*}} \right)= \frac{{\kappa \left( {{N_{\rm{B}}},w,s^{*}} \right)}}{{{C^2}\left( {{N_{\rm{B}}},w,s^{*}} \right)}},0 \le {P_{\rm{s}}}\left( {{N_{\rm{B}}},w,s^{*}} \right) \le 1
 \end{equation}
where ${{C^2}\left( {{N_{\rm{B}}},w,s^{*}} \right)}$ denotes the total possibilities of codeword pair for which each codeword represents the one choice from one node, i.e. Bob or Ava.  ${\kappa \left( {{N_{\rm{B}}},w,s^{*}} \right)}$ denotes  the number of codeword  pairs that  satisfy CS  when they overlap with each other.
\end{definition}
In this definition, we should note that  ${\kappa \left( {{N_{\rm{B}}},w,s^{*}} \right)}$  relies on  a fundamental fact:
\begin{fact}
1) The number  of zero digits in each codeword  determines how frequency CS can be broken down; 2) Those  zero digits, with uniform  spacing,  incur the most severe interference on CIR estimation accuracy.
 \end{fact}
This fact  also determines  why the instability of CIR estimation could occur.   We define the  Optimal Stability (OS) condition  by:
  \begin{definition}
There always exists ${P_{\rm{s}}}\left( {{N_{\rm{B}}},w,s^{*}} \right) = 1$  under arbitrary  SSAPs.
  \end{definition}
\subsubsection{Low-$N_{\rm B}$ scenario}
\begin{figure*}[!t]
\centering \includegraphics[width=1.00\linewidth]{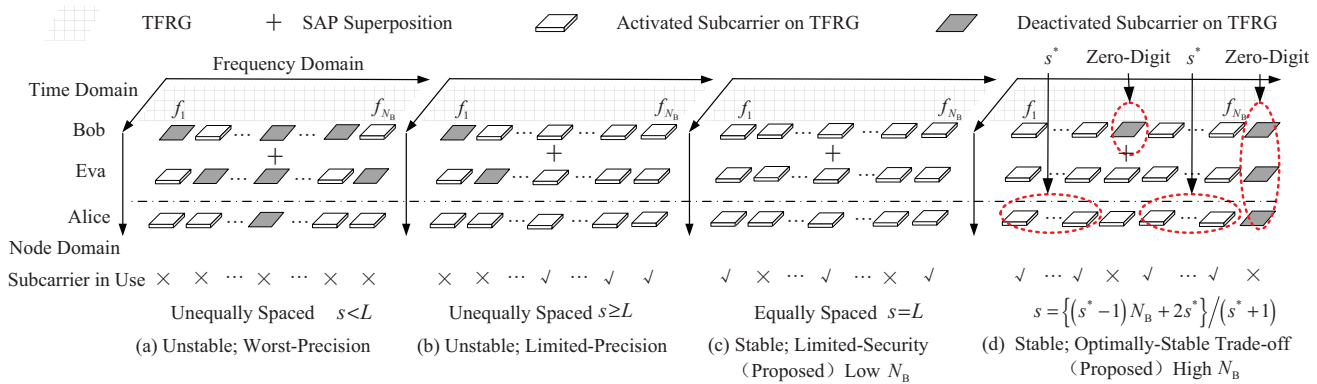}
\caption{The utilization strategy of subcarriers for CIR estimation and the analysis of estimation accuracy and stability. (a)~Influence of  the number of  overlapping subcarriers on CIR estimation; (b)~Influence of  the distribution of  overlapping subcarriers on CIR estimation; (c)~Example of   SSAPs  with highest stability but lowest security under low $N_{\rm B}$  (d)~Example of  SSAPs under high  $N_{\rm B}$ with optimally-stable estimation. This example is a worst case where the zero digits  are equally spaced  to destroy available subcarriers for stable CIR estimation. $s^{*}$ satisfies the Eq.~(\ref{E.44}) and $s$ satisfies $s = \frac{{\left( {{s^*} - 1} \right){N_{\rm{B}}} + 2{s^*}}}{{{s^*} + 1}}$.}
\label{PSOM2}
\vspace{-10pt}
\end{figure*}
Without loss of generality, we consider the low-$N_{\rm B}$ scenario where  $N_{\rm B}$ is equal to $s^{*}$.
Obviously, CS is satisfied when ${{\cal P}}_{s}$ is equal to the set ${\overline{\cal P}}_{s}$. In this case, we derive  the expression of  instability, defined by
\begin{equation}\label{E.46}
S_{T}\left( {{N_{\rm{B}}},w,{s^*}} \right) = {\left\{ {{{\left( {\begin{array}{*{20}{c}}
{{N_{\rm{B}}}}\\
{{N_{\rm{B}}} - w}
\end{array}} \right)} \mathord{\left/
 {\vphantom {{\left( {\begin{array}{*{20}{c}}
{{N_{\rm{B}}}}\\
{{N_{\rm{B}}} - w}
\end{array}} \right)} {\left( {\begin{array}{*{20}{c}}
{{N_{\rm{B}}} - {s^*}}\\
{{N_{\rm{B}}} - w}
\end{array}} \right){\rm{ }}}}} \right.
 \kern-\nulldelimiterspace} {\left( {\begin{array}{*{20}{c}}
{{N_{\rm{B}}} - {s^*}}\\
{{N_{\rm{B}}} - w}
\end{array}} \right){\rm{ }}}}} \right\}^2}
\end{equation}
with $s^{*}\le w \le N_{\rm B}\le \overline N $.

Based on this equation, we could  characterize the  relationship between the  security (defined by $S_{E}$ equal to ${1 \mathord{\left/
 {\vphantom {1 {{P_{\rm{I}}}}}} \right.
 \kern-\nulldelimiterspace} {{P_{\rm{I}}}}}$) and instability (i.e., $S_{T}$) as a fundamental  tradeoff existing in the whole uplink training process:
 \begin{fact}[\textbf{A Realistic  Tradeoff}]
\emph {The lower code rate brings the lower instability (Eq.~(\ref{E.46})); However,  the lower code rate causes the  higher security (Theorem 2 and Theorem 5). }
\end{fact}
 \begin{remark}
 For a mean AoA model with CPD, the tradeoff does not exist since $P_{\rm I}$ is always zero and thus  independent with the stability of CIR estimation. However, this is not realistic since  the mean AoA  is discretely distributed in practical scenarios with limited clusters. In this sense, the security-stability tradeoff is necessary and inevitable.
  \end{remark}
 The drawback of low-$N_{\rm B}$ configuration is that there is no security  when  Alice expects to achieve  OS condition and thus  $w$ should be  equal to $N_{\rm B}$ according to Eq.~(\ref{E.46}).  In other words, the tradeoff under OS condition cannot provide  desirable security guarantee  when  $N_{\rm B}$ is low. See the example  in Fig.~\ref{PSOM2}~(c).

We always expect to  maximize  the lower bound of security  by jointly optimizing  $N_{\rm B}$ and  $w$.  This object motives us  to  turn to large-$N_{\rm B}$ case.

\subsubsection{{High-$N_{\rm B}$ scenario and Optimally-Stable Tradeoff }}

In this part, we aim to determine the optimal ${R_s}$ such that  the security is maximized while the \textbf{OS}  condition is satisfied. Maximizing security means maximizing the code rate  since the security is a monotonic increasing function of code rate ${R_s}$. The optimization problem, also namely \textbf{Optimally-Stable Tradeoff}  problem,  can be formulated  by:
 \begin{equation}
\begin{aligned}
&\text{\,} \underset{{N_{\rm{B}}},w}{\text{max}}\quad{R_{ICC}}\left( {{N_{\rm{B}}},w} \right)\\
&s.t.\quad
{P_{\rm{s}}}\left( {{N_{\rm{B}}},w,s^{*}} \right) = 1,  s^{*} =\frac{{L - 1}}{L}N + 1
\end{aligned}
 \end{equation}
Before solving this problem, we need to fully understand  ${P_{\rm{s}}}\left( {{N_{\rm{B}}},w,s^{*}} \right) = 1$ under high-${N_{\rm{B}}}$. According to Fact 2 and  Fig.~\ref{PSOM2}~(d), we  have the following propositon:
\begin{proposition}
\textbf{OS} condition is satisfied iff  the number of adjacent non-zero digits  between any adjacent zero digits is at least equal to $s^{*}$ when zero digits are equally spaced for each of ICC codeword. We say  this   is named as the \textbf{$s^{*}$-OS} condition.
\end{proposition}
Inspired by this, we should optimize  ${N_{\rm B}}$ and $w$ such that  the non-zero digits are constrained  to create the \textbf{$s^{*}$-OS} condition. Under \textbf{$s^{*}$-OS} condition,  ${N_{\rm B}}$ should always satisfy
\begin{equation}\label{E.47}
\left( {{s^*} + 1} \right)\left( {{N_{\rm{B}}} - w} \right) + {s^*} \le {N_{\rm{B}}} \le \overline N, {s^*} \le w  \le {N_{\rm{B}}}\le \overline N
\end{equation}
The weight $w$ of  ICC-$\left( {N_{\rm B}, s} \right)$ should therefore  satisfy $\frac{{{s^*}}}{{{s^*} + 1}}\left( {{N_{\rm{B}}} + 1} \right) \le w \le {N_{\rm{B}}} \le \overline N$.
Especially, when $w$ is equal to ${N_{\rm{B}}}$, we have ${s^*} = w$. This corresponds to the low-$N_{\rm{B}}$ case.
\begin{figure*}[!t]
{
  \includegraphics[width=2.55in]{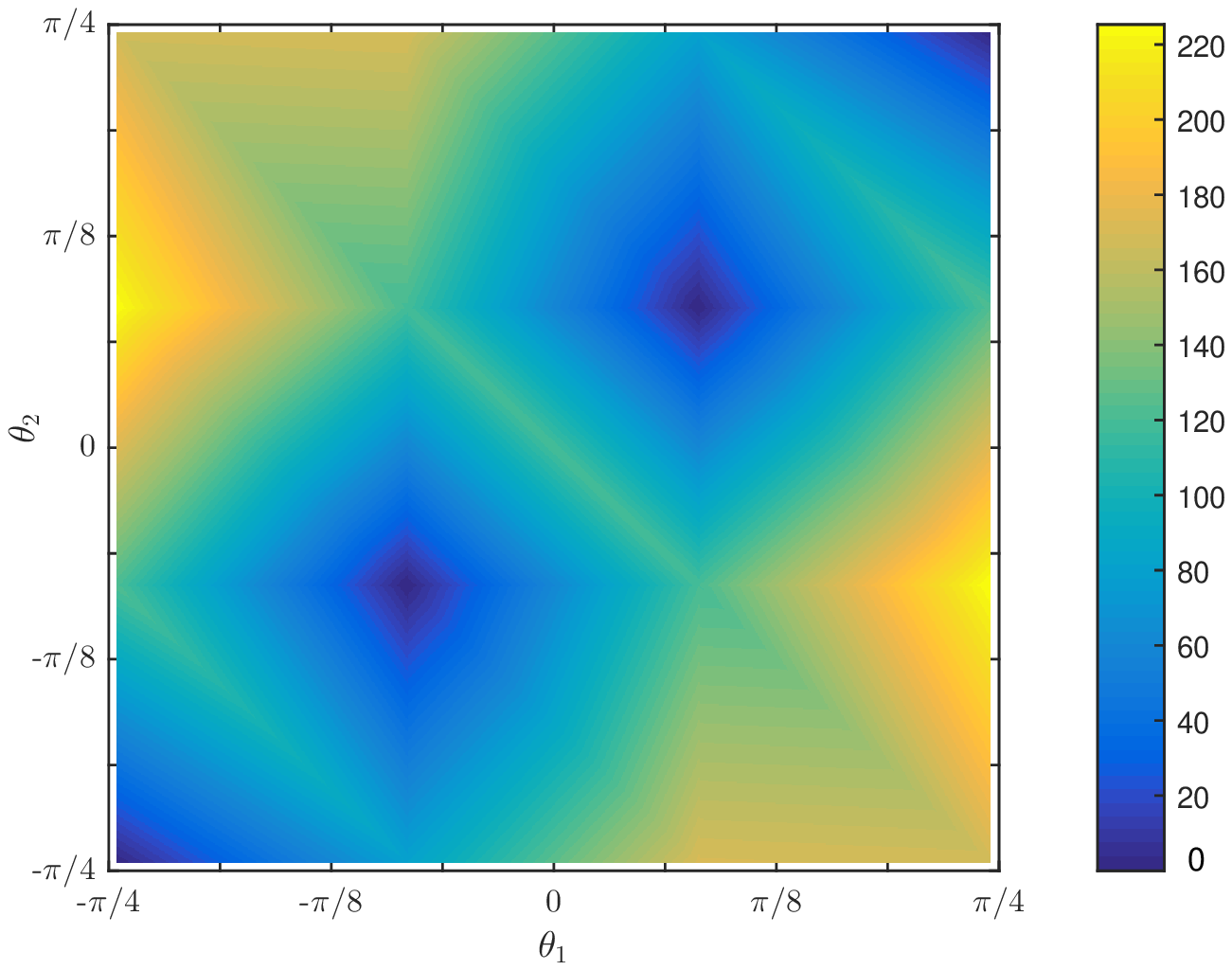} \hspace{-18pt} \includegraphics[width=2.55in]{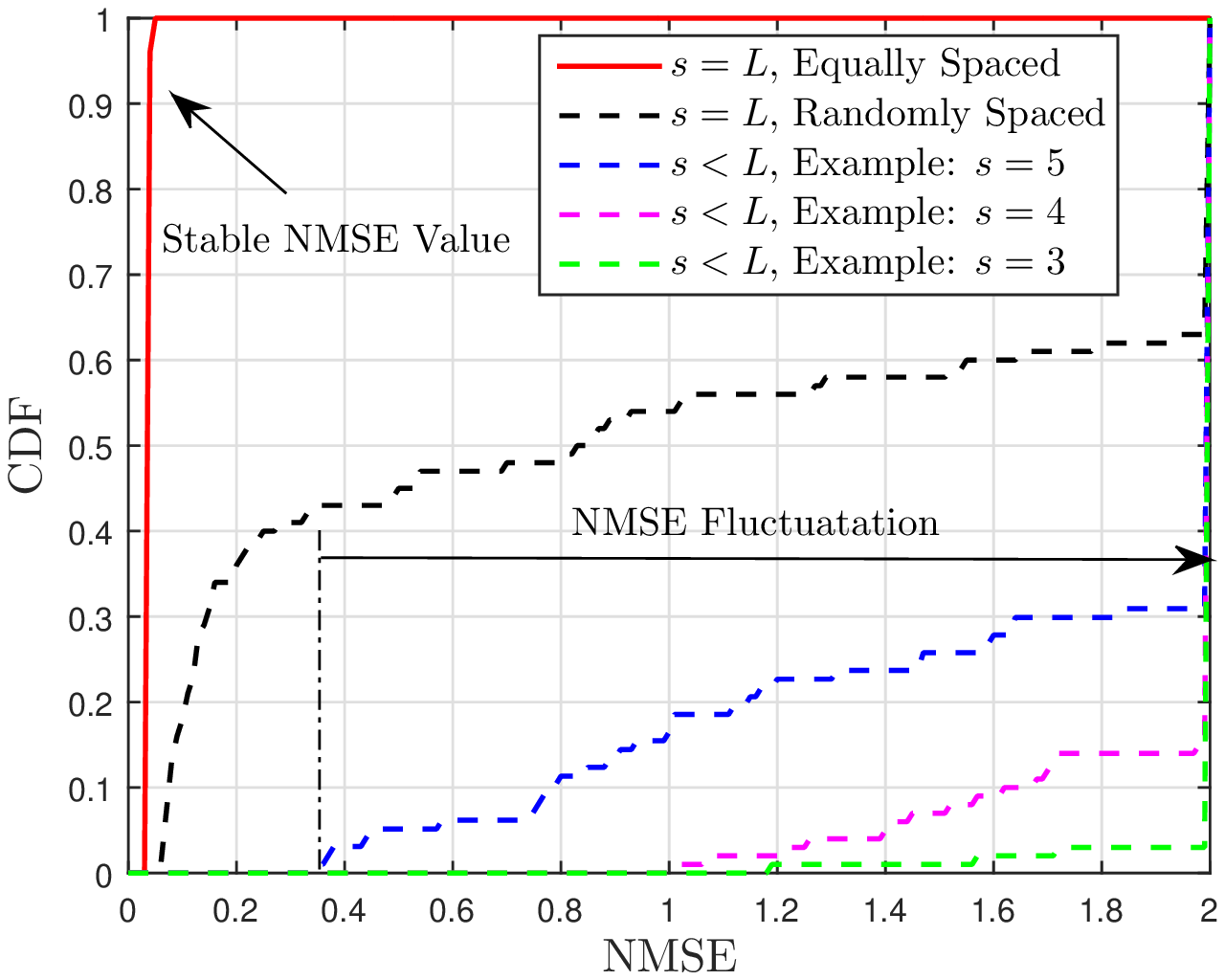}\hspace{-13pt}  \includegraphics[width=2.55in]{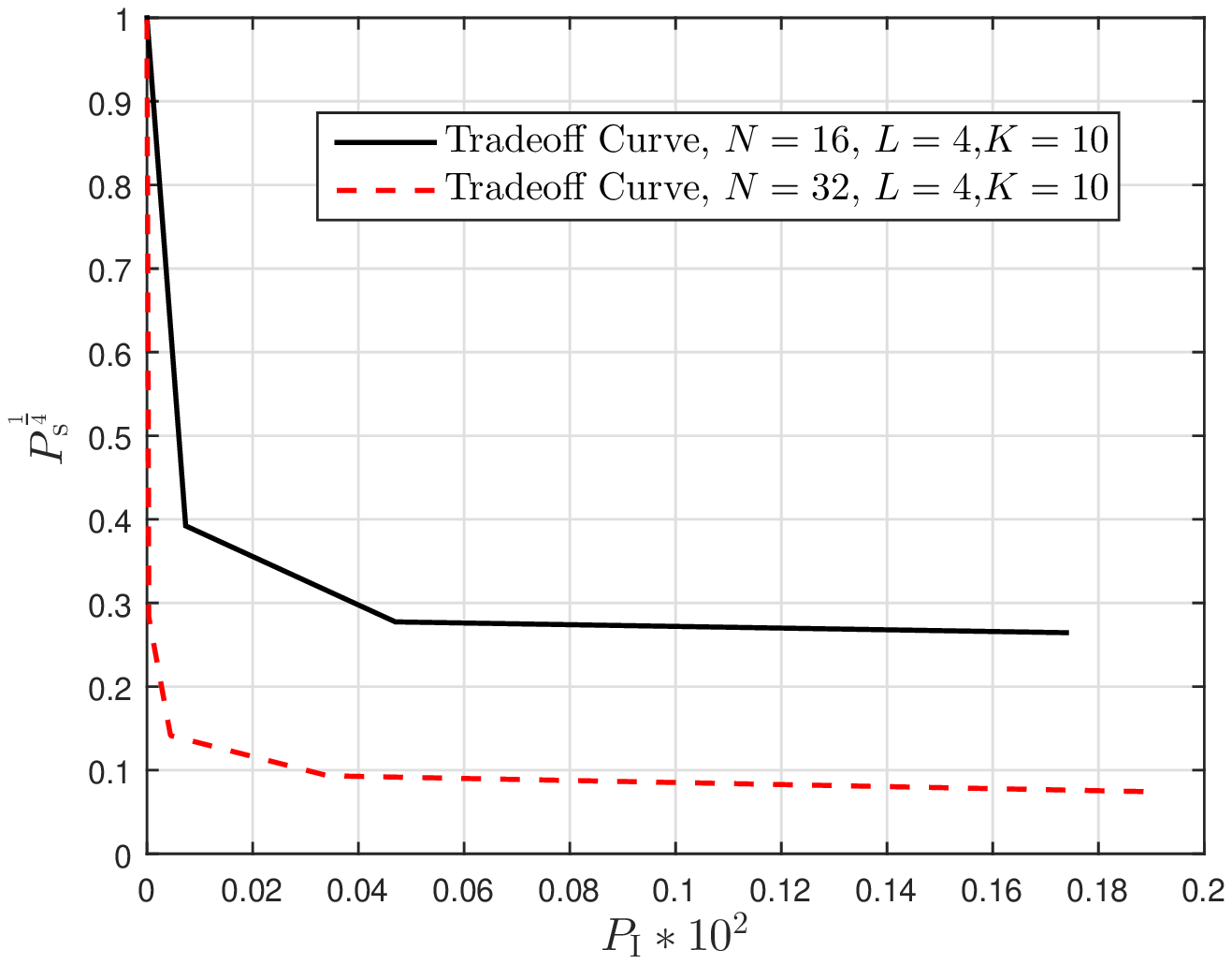}}
  \centerline{(a)\hspace{160pt}(b)\hspace{160pt}(c)}\vspace{-5pt}
  \caption{(a) Strength  of $\Delta f$ versus $\theta_{i},i=1,2$ with $N_{\rm T}=100$; (b) CDF of NMSE  under various SSAPs; (c) Security-instability tradeoff curve  with mean AoA discretely and  uniformly  distributed in a length-$K$ interval.}
  \label{Simulations1}
\end{figure*}

In this way, the \textbf{$s^{*}$-OS} condition is represented by the Eq.~(\ref{E.47}). And the maximization operation should be constrained by this equation.
\begin{theorem}
The optimal code rate  maximizing  the security while maintaining  the \textbf{$s^{*}$-OS} condition  can be calculated by
\begin{equation}\label{E.49}
{R_{\rm s}\left( {{N_{\rm{B}}},w,s^{*}} \right)} = {\log _2}{\left[{\frac{{{N_{\rm{B}}}!}}{{\left( {\frac{{{s^*}\left( {{N_{\rm{B}}} + 1} \right)}}{{{s^*} + 1}}} \right)!\left( {\frac{{{N_{\rm{B}}} - {s^*}}}{{{s^*} + 1}}} \right)!}}}\right]^{{1 \mathord{\left/
 {\vphantom {1 {{N_{\rm B}}}}} \right.
 \kern-\nulldelimiterspace} {{N_{\rm B}}}}}}
\end{equation}
The weight and order of optimally-stable code satisfy $w = \frac{{{s^*}}}{{{s^*} + 1}}\left( {{N_{\rm{B}}} + 1} \right)$ and $s = \frac{{\left( {{s^*} - 1} \right){N_{\rm{B}}} + 2{s^*}}}{{{s^*} + 1}}$.
\end{theorem}
\begin{proof}
See proof in Appendix~\ref{Theorem8}.
\end{proof}
By exploiting the property that there exists $\left( {\begin{array}{*{20}{c}}
n\\
k
\end{array}} \right) \ge {{{n^k}} \mathord{\left/
 {\vphantom {{{n^k}} {{k^k}}}} \right.
 \kern-\nulldelimiterspace} {{k^k}}}$ for all values of $n$ and $k$,  the lower bound approximation  of optimally-stable  ICC-$\left( {N_{\rm B}, s} \right)$ code  can be given by:
\begin{equation}\label{E.49.6}
{R_s}\left( {{N_{\rm{B}}},w,{s^*}} \right) \ge \frac{{{{\log }_2}\eta }}{\eta }
\end{equation}
with $\eta {\rm{ = }}\frac{{\left( {L - 1} \right)N + 2L}}{{\left( {L - 1} \right)N + L}}\frac{{{N_{\rm{B}}}}}{{{N_{\rm{B}}} + 1}}$.
\section{Numerical Results}
\label{NR}

In this section,  numerical simulations are presented to evaluate  above-mentioned techniques  during the  CTA process.

\subsection{ Numerical Verification for Theorem 4}
We  confirm the feasibility of Theorem 4  in Fig.~\ref{Simulations1}~(a) where  the strength  of $\Delta f $ is plotted  against  $\theta_{i}, i=1,2$ by configuring $N_{\rm T}=100$ and  $K=5$. To be more specific, the examples of   $\Delta f $ are derived from  the estimated FS channels and the correlation model in Eq.~(\ref{E.5}). $\theta_{i}, i=1,2$  are assumed to lie within the  set $\left\{ { - \frac{\pi }{4}, - \frac{\pi }{7},0, - \frac{\pi }{7}, - \frac{\pi }{4}} \right\}$.   As we can  see,  the identification error happens when $\Delta f =0$, that is, $\theta_{1}=\theta_{2}$.  In this sense, we  verified the feasibility of Theorem 4 and  could envision  that the IEP is zero under the assumption of the mean AoA with CPD.

\subsection{Security-Instability Tradeoff Curve}
In this subsection, we focus on the trade-off related results. We  evaluate in Fig.~\ref{Simulations1}~(b) the  fluctuation of NMSE employing  ICC-$\left( {N_{\rm B},s} \right)$ code  under various SSAPs, and then show how the security-instability tradeoff is developed  in Fig.~\ref{Simulations1}~(c).

In Fig.~\ref{Simulations1}~(b), we take the cumulative distribution function (CDF) of NMSE as the evaluation matric. The simulation is averaged over 100 runs, each of which perform 1000 channel average. We further consider that $N_{\rm B} =128$ are provided and  at most $s=L= 6$  subcarriers  overlap  for channel estimation. As a benchmark for  measuring  the instability, we simulate the ideal case where  six overlapping subcarriers  are always right selected.   As we can see, the CDF of NMSE under this ideal case is always stable. However, in practice,  ICC-$\left( {128,6} \right)$ code causes an undesirable status  where  the phenomenon of less-overlapping and  unequally-spaced subcarriers  occurs inevitably. This induces significant fluctuations of NMSE. As a consequence,  we present in Fig.~\ref{Simulations1}~(c)  the possibility of tradeoff between the security and  instability  by using parameters  $P_{\rm{I}}*10^{2}$ and  $P_{\rm{s}}^{\frac{1}{4}}$.  We consider  $N_{\rm B} =s^{*}=\frac{{L - 1}}{L}N + 1$ where  the FFT points $N$ is set to be either 16 or 32 while $L$ and $K$ are respectively fixed to be 4 and 10. As we can see, there exists a tradeoff curve on which the security has to be sacrificed to maintain a certain level of stability.
\begin{figure*}[!t]
{
  \includegraphics[width=2.55in]{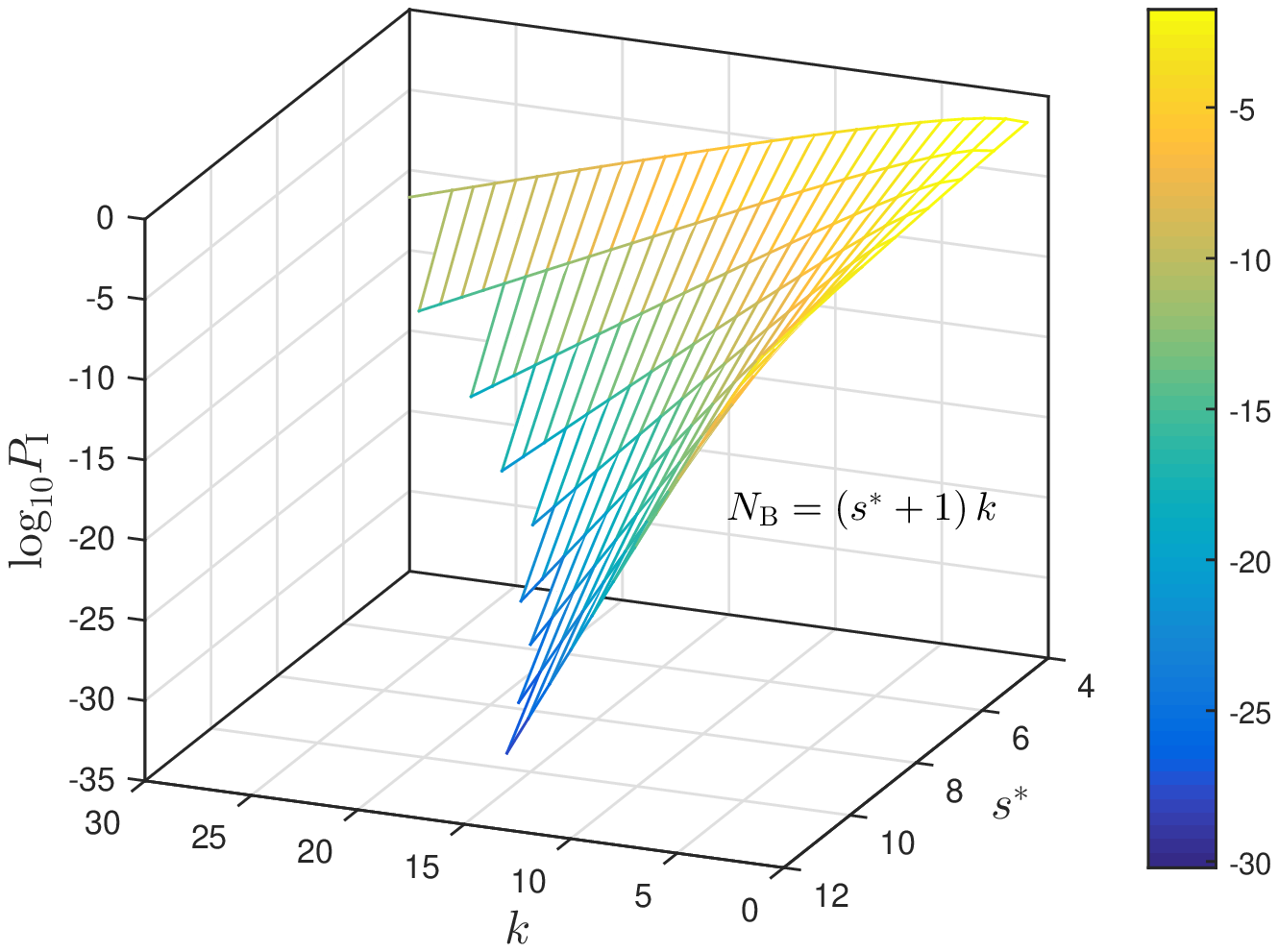} \hspace{-18pt} \includegraphics[width=2.55in]{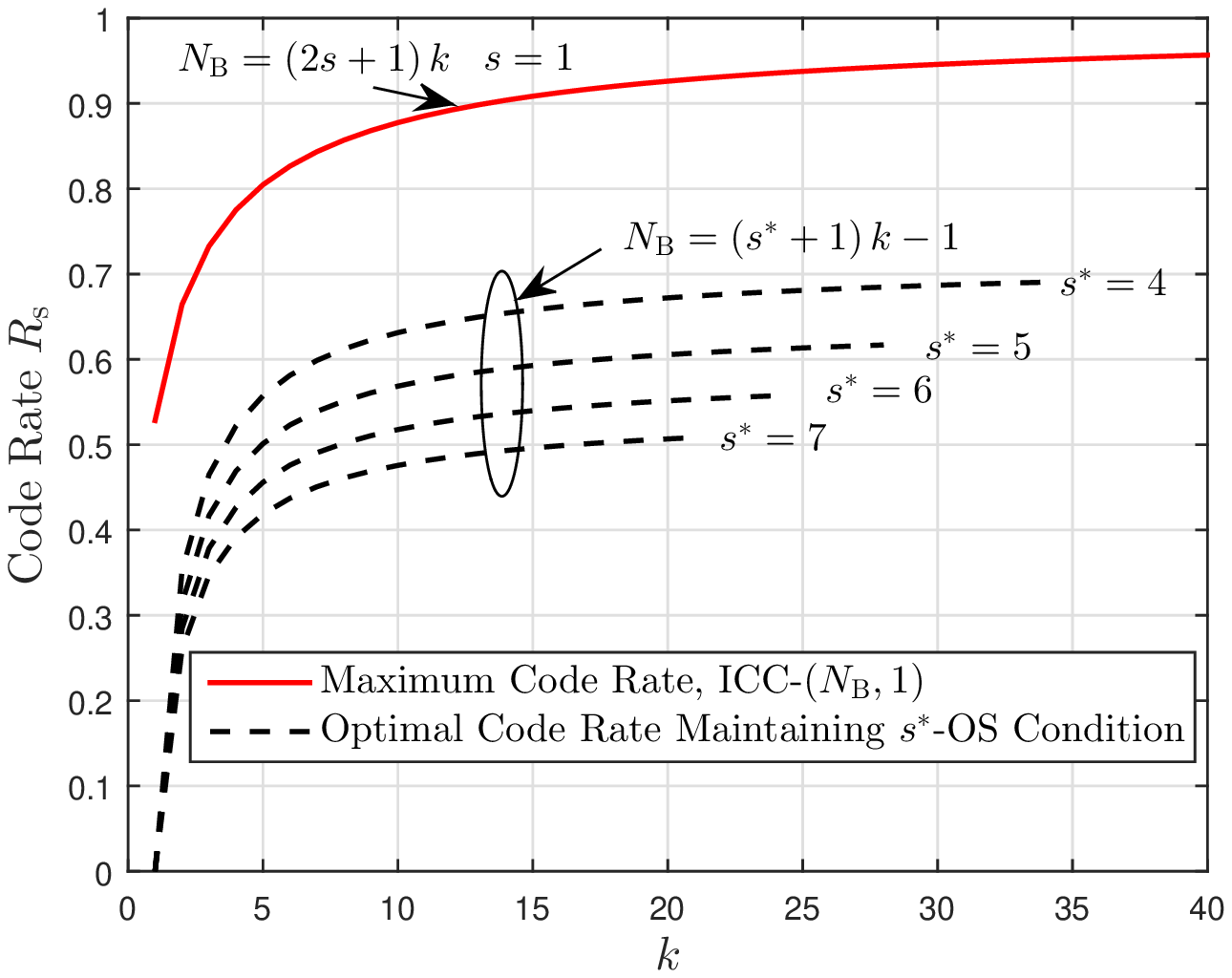}\hspace{-13pt}  \includegraphics[width=2.55in]{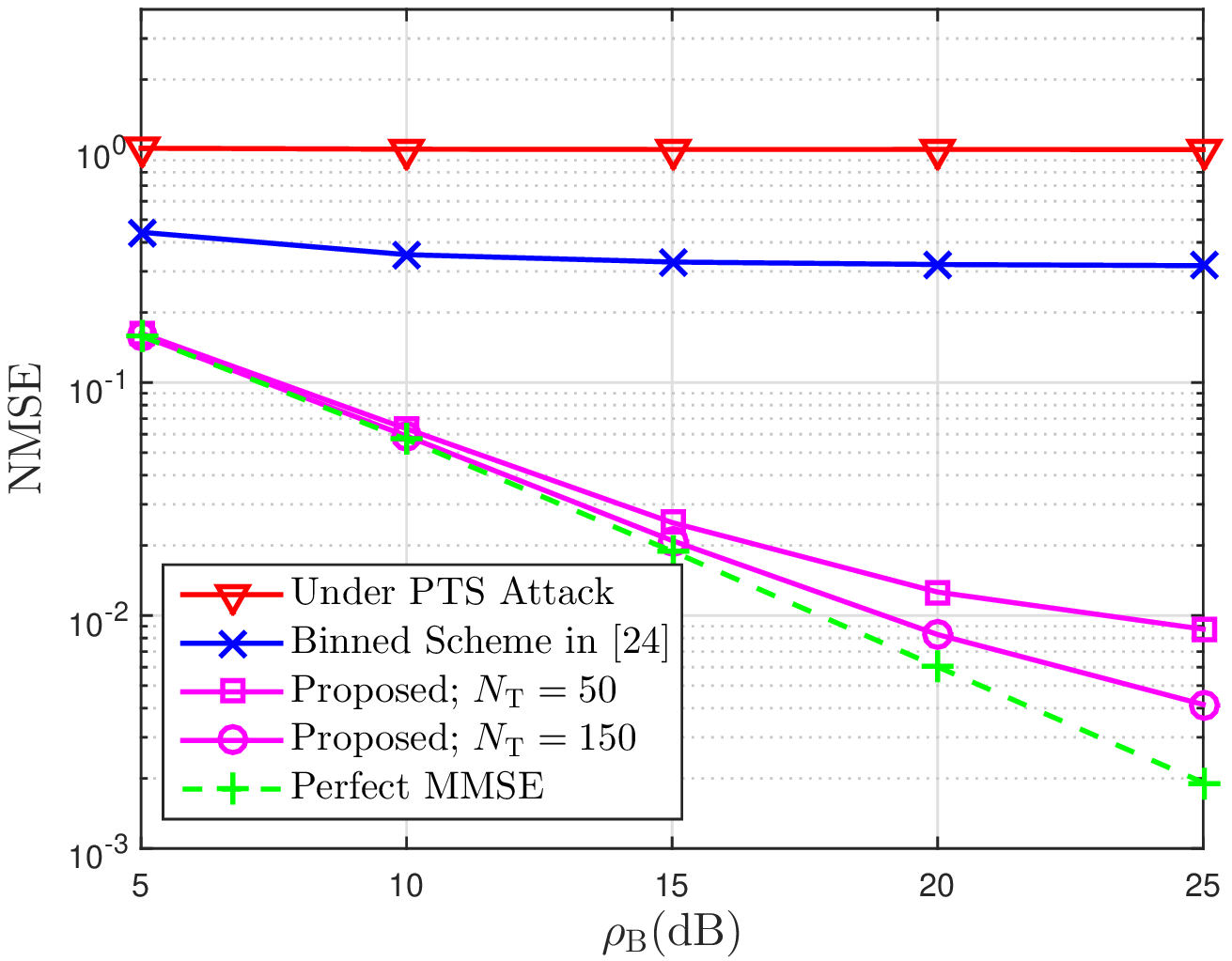}}
  \centerline{(a)\hspace{160pt}(b)\hspace{160pt}(c)}\vspace{-5pt}
  \caption{(a) Performance of IEP versus $N_{\rm B}$ and $s^{*}$ with mean AoA discretely and  uniformly  distributed in a length-$K$ interval; (b) Code rate versus $k$ under $s^{*}=\frac{{L - 1}}{L}N + 1$; (c) NMSE versus SNR of Bob under different number of antennas.}
  \label{Simulations2}
  \vspace{-15pt}
\end{figure*}
\subsection{ Security Under Optimally-Stable Tradeoff}
For this  part,  we should note that the IEP is zero  under the assumption of mean AoA obeying  CPD.  We consider the DPD model for the sake of practical analysis, and further simulate the IEP performance corresponding to the optimally-stable tradeoff in Fig.~\ref{Simulations2}~(a).
In this figure,  the 3D  plot of  IEP is sketched  versus $N_{\rm B}$ and $s^{*}$.   We consider $s^{*}$ to be from 4 to 12 and $K$ to be 20. $k$, related to $N_{\rm B}$, satisfies $N_{\rm B}+1= \left( {s^{*} + 1} \right)k$.    As we can see,  IEP decreases  with the increase of $N_{\rm B}$ and $s^{*}$.  On one hand, the initial value of $s^{*}$ determines how fast  the IEP can decrease and what is the  minimum  value  IEP can achieve. For example,  IEP decreases faster  with the increase of $s^{*}$, and $P_{\rm I}$  achieves as low as $10^{-3}$ at $k=15$ when $s^{*}$ is equal to 12. In this case, the number of occupied subcarriers is required to be  $N_{\rm B}=195$.  On the other hand, the initial value of $N_{\rm B}$ also determines  the tendency  for  the variable  $s^{*}$ to be reduced.  Specially, at a large $N_{\rm B}$,  a decreasing function $P_{\rm I}$ of $s^{*}$, at least within the interval $\left[ {4,12} \right]$,  can be created.
\vspace{-5pt}
\subsection{ Code Rate Under Optimally-Stable Tradeoff}
In Fig.~\ref{Simulations2}~(b), we  evaluate the code rate under the optimally-stable  tradeoff. Before that, we  consider the Eq.~(\ref{E.21}) for comparison and sketch the  curve of maximum code rate under $s=1$  over  $k$.  On this reference curve,  the code rate increases  and gradually approach 1 with the increase of $k$. As to the optimally-stable  tradeoff,  we simulate the curves of code rate shown in Eq.~(\ref{E.49}) over  $s^{*}$  from  4 to 7.  As we can see, the code rate in this case  is reduced compared with   that without tradeoff consideration. With the increase of  $s^{*}$, we have to  get less code rate. For example,  the code rate under $s^{*}=7, k=21$ and thus $N_{\rm B}=167$ is equal to 0.5083, which  means the rate loss of 0.4205 (almost 45 percent) is caused by the tradeoff at this point.
\vspace{-5pt}
\subsection{CIR Estimation   Under Optimally-Stable Tradeoff}
Finally, we stimulate the performance  of stable CIR estimation in Fig.~\ref{Simulations2}~(c) where  the  NMSE  is presented versus SNR of Bob under different number of antennas.  $L$ and $N_{\rm B}$ are respectively  configured to be 6 and  256. Here,  we  consider the   estimation using Eq.~(\ref{E.27.5})  and assume perfect identification  for attacks.  The performance under  this type of estimator  is not influenced by the specific value of $\rho_{\rm A}$ due to the subspace projection property. We configure $\rho_{\rm A}=\rho_{\rm B}$ and  do not consider the case where there is no attack since in this case LS estimator is a natural choice.  For the simplicity of  comparison, we only present the channel estimation  under PTS attack because  the estimation error floor  under PTN and PTJ attack can be easily understood to be very high. The binned scheme proposed in~\cite{Shahriar2} is simulated as an another  comparison scheme.   As we can see,  PTS attack causes a high-NMSE floor on  CIR estimation  for Bob. This phenomenon can also be seen in the binned scheme. However, the estimation in our proposed framework  breaks down this floor and its NMSE  gradually decreases with the increase of transmitting antennas.  Also, we consider perfect MMSE to be  a performance benchmark for which  perfect pilot tones, including Ava's pilot tones,  are assumed to be  known by Alice. We find that  the NMSE brought in our scheme  gradually approaches the level under perfect MMSE with the increase of antennas. That's because the asymptotically-optimal  estimator highly relies on the statistical covariance matrix which is determined by the number of antennas.
\section{Conclusions}
\label{Conclusions}
This paper investigated the issue of pilot-aware attack on the uplink CTA in large-scale MISO-OFDM systems. We proposed a secure ICC-CTA protocol in which  pilot tones, usually exposed in public, are now enabled  to be shared between legitimate transceiver pair, with high security under hybrid attack environment. Theoretically, we discovered an critical fact that this architecture could exhibit a perfect  security if the CPD  model of mean AoA was considered. In practical scenarios with the DPD model of mean AoA, this architecture  was required to make tradeoff between the security   and stability of CIR estimation. We  showed that given a suitable code rate, stable CIR estimation  could be always maintained  under a high security.
We conclude this paper by pointing out some interesting topics for future work. As one interesting direction,  more delicate optimization on the tradeoff could be further researched such that the code rate under optimally-stable tradeoff could be higher. The extension to solving the issue of  pilot contamination in massive MIMO systems could be another interesting direction since the  pilot phases guaranteed by our scheme  can be superimposed onto the traditional optimized  pilots and  thus control even avoid pilot contamination in multi-cell scenarios with only three OFDM symbol time.
\section{Appendix}
\subsection{Proof of Theorem 1}
\label{Theorem1}
Since codewords in this constant-weight code are constrained to be  with same and  fixed length,  the number of overlapping digits achieves its minimum only when the zero digits of each codeword  are  fully occupied. In this case, the remanent digits, i.e., the overlapping digits,  account for  ${2w-N_{\rm B}}$ which should be  equal to $s$ and less than $w$.  Therefore, we can prove the theorem.
\subsection{Proof of Theorem 2}
\label{Theorem2}
Considering  the hybrid attack,  we know that there exists the  possibility  of  $2^{N_{\rm B}}$   codewords to appear.  Two interpreted codewords derived under  ${\cal A}_{1}$ and  ${\cal A}_{0}$,  if satisfying  $N_1^{\rm d}+N_{1,1}^{\rm s}=N_1^{\rm d}+N_{1,0}^{\rm s}$, will confuse Alice. In this case,  each assumption is  decided with the probability of $0.5$.   The possible number of codewords that satisfy  this condition is equal to   ${ {\frac{{{N_{\rm B}}!}}{{\left( {\frac{{{N_{\rm B}} + s}}{2}} \right)!\left( {\frac{{{N_{\rm B}} - s}}{2}} \right)!}}}}$.  One exception is when  the codeword of Ava is identical to that of Bob. In this case, the codeword can be uniquely identified.  Finally,  there exists the  possibility of   ${ {\frac{{{N_{\rm B}}!}}{{\left( {\frac{{{N_{\rm B}} + s}}{2}} \right)!\left( {\frac{{{N_{\rm B}} - s}}{2}} \right)!}}}-1}$ codewords that could cause identification errors. Then the ultimate IEP can be proved.
\subsection{Proof of Proposition 3}
\label{Proposition3}
Taking Bob for example, we can derive the estimation error as $\varepsilon _{\rm{B}}^2 = T\left( {1 - T{\bf{x}}_{{\rm{L}},{\rm{1}}}^{\rm{H}}{\bf{C}}_{{{\bf{Y}}_{\rm{L}}}}^{ - 1}{{\bf{x}}_{{\rm{L}},{\rm{1}}}}} \right)$. Now let us focus on the term ${{\bf{C}}_{{{\bf{Y}}_{\rm L}}}}$. We can express  ${{\bf{h}}_{{\rm{B}},{\rm{L}}}}$ as ${{\bf{g}}_{{\rm{B}},{\rm{L}}}}\left( {{\bf{R}}_{\rm{1}}^{{1 \mathord{\left/
 {\vphantom {1 2}} \right.
 \kern-\nulldelimiterspace} 2}} \otimes  {\bf{F}}_{{\rm{L}},{{s}}}^{\rm{T}}} \right)$ where ${{\bf{g}}_{{\rm{B}},{\rm{L}}}} $ is the integrated $1 \times N_{\rm T}L$  CIR vector  of i.i.d.  ${\cal C}{\cal N}\left( {0,{1}} \right)$ random variables. Based on the \emph{Lemma B.26} in~\cite{Hoydis},
 ${{\bf{C}}_{{{\bf{Y}}_{\rm L}}}}$  is  then transformed into ${{\bf{C}}_{{{\bf{Y}}_{\rm L}}}}  \xlongrightarrow[{N_{\rm T}} \to \infty]{ \rm{a.s.}} \frac{1}{{{N_{\rm{T}}}s}}{{\bf{X}}_{{\rm{L}}}}{{\bf{R}}_{\rm C}}{\bf{X}}_{{\rm{L}}}^{\rm{H}} + {\sigma ^2}{{\bf{I}}_2}$.
Here,  the $2\times 2$  matrix ${{\bf{R}}_{\rm C}}$  satisfies ${{\bf{R}}_{\rm{C}}} = {\rm{diag}}\left\{ {{{\left[ {\begin{array}{*{20}{c}}
{{\rm{Tr}}\left( {{{\bf{R}}_1}} \right){\rm{Tr}}\left( {{{\bf{R}}_{\rm{F}}}} \right)}&{{\rm{Tr}}\left( {{{\bf{R}}_2}} \right){\rm{Tr}}\left( {{{\bf{R}}_{\rm{F}}}} \right)}
\end{array}} \right]}^{\rm{T}}}} \right\}$.
Therefore, we can derive  $\varepsilon _{\rm{B}}^2 =T\left\{ {1 - {\bf{x}}_{{\rm{L}},{\rm{1}}}^{\rm{H}}{{\left( {{{\bf{X}}_{\rm{L}}}{\bf{X}}_{\rm{L}}^{\rm{H}}} \right)}^{ - 1}}{{\bf{x}}_{{\rm{L}},{\rm{1}}}}} \right\}$  at high SNR region.
In the same way, we can derive  $\varepsilon _{\rm{A}}^2 =T\left\{ {1 - {\bf{x}}_{{\rm{L}},{\rm{2}}}^{\rm{H}}{{\left( {{{\bf{X}}_{\rm{L}}}{\bf{X}}_{\rm{L}}^{\rm{H}}} \right)}^{ - 1}}{{\bf{x}}_{{\rm{L}},{\rm{2}}}}} \right\}$.
After calculating  the matrix inverse and performing matrix multiplication,  we can finally verify $ \varepsilon _{\rm{B}}^2 = \varepsilon _{\rm{A}}^2 $. This completes the proof.
\subsection{Proof of Theorem 3}
\label{Theorem3}
Thanks to ${\widehat {\bf{h}}_{\rm{B,L}}} = {{\bf{h}}_{\rm{B,L}}} - \varepsilon _{{\rm{B}}}{\bf{h}}$, the measure $ f\left( {{{\widehat {\bf{h}}}_{{\rm{B}},{\rm{L}}}}} \right)$ can be expressed as the equation satisfying  ${f\left( {{{\widehat {\bf{h}}}_{{\rm{B}},{\rm{L}}}}} \right) = \left( {{{\bf{h}}_{{\rm{B}},{\rm{L}}}} - {\varepsilon _{\rm{B}}}{\bf{h}}} \right)\left( {{{\overline {\bf{R}} }_1} \otimes {{\overline {\bf{R}} }_{\rm{F}}}} \right){{\left( {{{\bf{h}}_{{\rm{B}},{\rm{L}}}} - {\varepsilon _{\rm{B}}}{\bf{h}}} \right)}^{\rm{H}}}}$.
This equation can be expanded into$f\left( {{{\widehat {\bf{h}}}_{{\rm{B}},{\rm{L}}}}} \right) = {f_1} - 2{f_2} + {f_3}$
  with ${f_1} = {{\bf{h}}_{{\rm{B}},{\rm{L}}}}\left( {{{\overline {\bf{R}} }_1} \otimes {{\overline {\bf{R}} }_{\rm{F}}}} \right){\bf{h}}_{{\rm{B}},{\rm{L}}}^{\rm{H}}$, ${f_2} = {\varepsilon _{\rm{B}}}{{\bf{h}}_{{\rm{B}},{\rm{L}}}}\left( {{{\overline {\bf{R}} }_1} \otimes {{\overline {\bf{R}} }_{\rm{F}}}} \right){\bf{h}}$ and  ${f_3} = \varepsilon _{\rm{B}}^2{\bf{h}}\left( {{{\overline {\bf{R}} }_1} \otimes {{\overline {\bf{R}} }_{\rm{F}}}} \right){\bf{h}}$.
By using  the \emph{Lemma B.26} in~\cite{Hoydis} for each term, we can have $\frac{{f\left( {{{\widehat {\bf{h}}}_{{\rm{B}},{\rm{L}}}}} \right)}}{{{N_{\rm{T}}}s}}  \xlongrightarrow[{N_{\rm T}} \to \infty]{ \rm{a.s.}}  \frac{{{\rho _1}L{\rm{ + }}\varepsilon _{\rm{B}}^2{\rm{Tr}}\left( {{{\overline {\bf{R}} }_1} \otimes {{\overline {\bf{R}} }_{\rm{F}}}} \right)}}{{{N_{\rm{T}}}s}}$.
In the same way, we can  obtain the relationship $\frac{{f\left( {{{\widehat {\bf{h}}}_{{\rm{A}},{\rm{L}}}}} \right)}}{{{N_{\rm{T}}}s}} \xlongrightarrow[{N_{\rm T}} \to \infty]{ \rm{a.s.}}  \frac{{L{\rm{Tr}}\left( {{{\bf{R}}_2}{{\overline {\bf{R}} }_1}} \right){\rm{ + }}\varepsilon _{\rm{A}}^2{\rm{Tr}}\left( {{{\overline {\bf{R}} }_1} \otimes {{\overline {\bf{R}} }_{\rm{F}}}} \right)}}{{{N_{\rm{T}}}s}}$.
 As indicated in Proposition 3, there exists $\varepsilon _{\rm{B}}^2=\varepsilon _{\rm{A}}^2$. By comparing the two simplified results of $f\left( {{{\widehat {\bf{h}}}_{{\rm{B}},{\rm{L}}}}} \right)$ and $f\left( {{{\widehat {\bf{h}}}_{{\rm{A}},{\rm{L}}}}} \right)$, we  can complete the proof.
 \subsection{Proof of Theorem 4}
 \label{Theorem4}
 First, we will prove $\sum\limits_{j = 1}^a {\frac{{{\lambda _{2,{i_j}}}}}{{{\lambda _{1,{i_j}}}}}}  = a$. As shown in~\cite{Adhikary}, the empirical CDF of eigenvalues of ${\bf {R}}_{i}$ can be asymptotically  approximated  by the samples from  $\left\{ {{S_i}\left( {\left[ {{n \mathord{\left/
 {\vphantom {n {{N_{\rm{T}}}}}} \right.
 \kern-\nulldelimiterspace} {{N_{\rm{T}}}}}} \right]} \right),n = 0, \ldots ,{N_{\rm{T}}} - 1} \right\}$. Therefore, the eigenvalues of different individuals, if overlapping at the same location, e.g.,  $n$, can be approximated with the same eigenvalue. In this case, the ratio of two eigenvalues at the same location is one and therefore, we can prove $\sum\limits_{j = 1}^a {\frac{{{\lambda _{2,{i_j}}}}}{{{\lambda _{1,{i_j}}}}}}  = a$ for $a$ overlapping positions.
 Then we prove that there must $a < {\rho _1}$. Examining $\left[ {{\theta _2} - {\Delta },{\theta _2} + {\Delta }} \right]$ and  $\left[ {{\theta _1} - {\Delta },{\theta _1} + {\Delta }} \right]$, we found that  if $ {\theta _1} \ne {\theta _2}$ is satisfied, there must exist $a < {\rho _1}$ since $\left[ {{\theta _2} - {\Delta },{\theta _2} + {\Delta }} \right]$ must have non-empty intersection  with  $\left[ {{\theta _1} - {\Delta },{\theta _1} + {\Delta}} \right]$.  In this case, the  number of elements in ${{\cal S}_{3}}$  is reduced to be smaller than that ${\rho _1}$.
 Now we turn to the case ${\theta _1} = {\theta _2}$ in which we easily have ${{\bf{R}}_1} = {{\bf{R}}_2}$ and therefore the theorem is proved.
  \subsection{Proof of Theorem 6}
 \label{Theorem8}
 Let  us determine  the value of minimum of $w$. From Eq.~(\ref{E.47}), we know that there exists $ w \ge {s^*}$ and $w \ge \frac{{{s^*}}}{{{s^*} + 1}}\left( {{N_{\rm{B}}} + 1} \right)$.  Since ${N_{\rm{B}}} \ge {s^*}$, we can acquire $w = \frac{{{s^*}}}{{{s^*} + 1}}\left( {{N_{\rm{B}}} + 1} \right)$ as the minimum of $w$. Note that it satisfies $w \ge \frac{{{N_{\rm B}} + 1}}{2}$ for $s^{*}> 1$. In this case, the value of $C$ will decrease with the increase of $w$. Thus the maximum code rate, i.e. maximum security, can be achieved  at this weight.  Moreover, according to the Theorem 1, we can know there exists $w = \frac{{{N_{\rm{B}}} + s}}{2}$ for an  ICC-$\left( {N_{\rm B}, s} \right)$ code and therefore we can derive the relationship between $s$ and $s^{*}$. The theorem is finally proved.

\end{document}